\documentclass[preprint,12pt]{elsarticle}
   
\makeatletter
\def\ps@pprintTitle{%
 \let\@oddhead\@empty
 \let\@evenhead\@empty
 \def\@oddfoot{}%
 \let\@evenfoot\@oddfoot}
\makeatother

\usepackage{amsmath,amsthm}
\usepackage{url}
\usepackage{enumerate}
\usepackage{xspace}
\usepackage{graphicx}
\usepackage{graphics}
\usepackage{proof}
\usepackage{algorithm}
\usepackage[noend]{algpseudocode}
\usepackage{multicol}

\newtheorem{exmp}{Example}

\usepackage{enumitem, color}
\setitemize{itemsep=0pt,topsep=4pt,parsep=4pt,partopsep=4pt}
\setenumerate{itemsep=0pt,topsep=4pt,parsep=4pt,partopsep=4pt}
\usepackage{tikz}
\usetikzlibrary{positioning}
\usetikzlibrary{arrows}
\usepackage{mathtools}
\usepackage{tikz}
\usetikzlibrary{arrows}
\usepackage{caption}
\usepackage{float}
\usepackage{comment}
\usepackage{enumerate}
\usepackage{booktabs}
\usepackage{thmtools,thm-restate}

\usepackage{environ}

\newcommand{\repeattheorem}[1]{%
  \begingroup
  \renewcommand{\thetheorem}{\ref{#1}}%
  \expandafter\expandafter\expandafter\theorem
  \csname reptheorem@#1\endcsname
  \endtheorem
  \endgroup
}

\NewEnviron{reptheorem}[1]{%
  \global\expandafter\xdef\csname reptheorem@#1\endcsname{%
    \unexpanded\expandafter{\BODY}%
  }%
  \expandafter\theorem\BODY\unskip\label{#1}\endtheorem
}

\usetikzlibrary{decorations.pathreplacing,calc}

\newcommand{\problemFont}[1]{\protect\ensuremath{\mathsf{#1}}\xspace}

\newcommand{\classFont}[1]{\protect\ensuremath{\mathsf{#1}}\xspace}
\newcommand{\coNP}{\classFont{co-NP}}

\newcommand{\sub}{\subseteq}
\newcommand{\tuple}[1]{\vec{#1}}

\newcommand{\Dom}{\textrm{Dom}}

\newcommand{\A}{\mathfrak{A}}

\newcommand{\N}{\mathbb{N}}

\newcommand\re[2]{{% we make the whole thing an ordinary symbol
  \left.\kern-\nulldelimiterspace % automatically resize the bar with \right
  #1 % the function
  \vphantom{\big|} % pretend it's a little taller at normal size
  \right|_{#2} % this is the delimiter
  }}

\newcommand{\calX}{\mathcal{X}}
\newcommand{\calY}{\mathcal{Y}}
\newcommand{\calZ}{\mathcal{Z}}
\newcommand{\calU}{\mathcal{U}}
\newcommand{\calV}{\mathcal{V}}
\newcommand{\calP}{\mathcal{P}}
\newcommand{\calK}{\mathcal{K}}

\newcommand{\threeSAT}{\problemFont{3}-\problemFont{SAT}}
\newcommand{\overbar}[1]{\mkern 1.4mu\overline{\mkern-1.4mu#1\mkern-1.4mu}\mkern 1.4mu}

\newtheorem{theorem}{Theorem}
\newtheorem{lemma}[theorem]{Lemma}

\newtheorem{corollary}[theorem]{Corollary}

\newtheorem{definition}[theorem]{Definition}

\newtheorem{example}[theorem]{Example}

\usepackage{amssymb}
\journal{JCSS}

\begin{document}

\begin{frontmatter}



\title{Controlling Entity Integrity with Key Sets\footnote{Some of our results were announced at the 23rd International Conference on Automated Reasoning (IJCAR 2018)}}


\author[Suomi]{Miika Hannula}%\footnote{Corresponding author}}
\ead{miika.hannula@helsinki.fi}
\author[Aotearoa]{Xinyi Li}
\ead{}
\author[Aotearoa]{Sebastian Link}
\ead{s.link@auckland.ac.nz}
\address[Suomi]{Department of Mathematics and Statistics, University of Helsinki, Helsinki, Finland}
\address[Aotearoa]{School of Computer Science, University of Auckland, New Zealand}

\begin{abstract}
Codd's rule of entity integrity stipulates that every table has a primary key. Hence, the attributes of the primary key carry unique and complete value combinations. In practice, data cannot always meet such requirements. Previous work proposed the superior notion of key sets for controlling entity integrity. We establish a linear-time algorithm for validating whether a given key set holds on a given data set, and demonstrate its efficiency on real-world data. We establish a binary axiomatization for the associated implication problem, and prove its coNP-completeness.  However, the implication of unary by arbitrary key sets has better properties. The fragment enjoys a unary axiomatization and is decidable in quadratic time. Hence, we can minimize overheads before validating key sets. While perfect models do not always exist in general, we show how to compute them for any instance of our fragment. This provides computational support towards the acquisition of key sets.
\end{abstract}

\begin{comment}
\begin{highlights}
\item A linear-time algorithm can validate key sets efficiently on real-world data.
\item The implication problem of key sets is shown to be coNP-complete.
\item Implication of key sets is equivalent to a propositional fragment over a three-valued logic.
\item Key sets enjoy a binary axiomatization but do not enjoy Armstrong relations.
\item Implication of unary by arbitrary key sets can be decided in linear time.
\item Implication of unary by arbitrary key sets enjoys a unary axiomatization and Armstrong relations.
\end{highlights}

\begin{keyword}
Axiomatization\sep Complexity\sep Entity integrity \sep Implication \sep Key \sep Key set\sep Missing data\sep Reasoning \sep Relational database\sep Validation
\end{keyword}
\end{comment}

\end{frontmatter}

\section{Introduction}

Keys provide efficient access to data in database systems. They are required to understand the structure and semantics of data. For a given collection of entities, a key refers to a set of column names whose values uniquely identify an entity in the collection. For example, a key for a relational table is a set of columns such that no two different rows have matching values in each of the key columns. Keys are fundamental for most data models, including semantic models, object models, XML, RDF, and graphs. They advance many classical areas of data management such as data modeling, database design, and query optimization. Knowledge about keys empowers us to 1) uniquely reference entities across data repositories, 2) reduce data redundancy at schema design time to process updates efficiently at run time, 3) improve selectivity estimates in query processing, 4) feed new access paths to query optimizers that can speed up the evaluation of queries, 5) access data more efficiently via physical optimization such as data partitioning or the creation of indexes and views, and 6) gain new insight into application data. Modern applications create even more demand for keys. Here, keys facilitate data integration, help detect duplicates and anomalies, guide the repair of data, and return consistent answers to queries over dirty data. The discovery of keys from data sets is a core task of data profiling.

Due to the demand in real-life applications, data models have been extended to accommodate missing information. The industry standard for data management, SQL, allows occurrences of a null marker to model any kind of missing value. Occurrences of the null marker mean that no information is available about an actual value of that row on that attribute, not even whether the value exists and is unknown nor whether the value does not exist. Codd's principle of entity integrity suggests that every entity should be uniquely identifiable. In SQL, this has led to the notion of a primary key. A primary key is a collection of attributes which stipulates uniqueness and completeness. That is, no row of a relation must have an occurrence of the null marker on any columns of the primary key and the combination of values on the columns of the primary key must be unique. The requirement to have a primary key over every table in the database is often inconvenient in practice. Indeed, it can happen easily that a given relation does not exhibit any primary key. This is illustrated by the following example.

\begin{example}\label{ex:intro} Consider the following snapshot of data from an accident ward at a hospital \cite{DBLP:journals/dam/Thalheim92}. Here, we collect information about
the \emph{name} and \emph{address} of a patient, who was treated for an \emph{injury} in some \emph{room} at some \emph{time}.
\begin{center}
\begin{tabular}{c@{\hspace*{.5cm}}c@{\hspace*{.5cm}}c@{\hspace*{.5cm}}c@{\hspace*{.5cm}}c}\hline
\textit{room} & \textit{name} & \textit{address} & \textit{injury} & \textit{time} \\ \hline
1             & Miller        & $\perp$          & cardiac infarct & Sunday, 19 \\
$\perp$       & $\perp$       & $\perp$          & skull fracture  & Monday, 19 \\
2             & Maier         & Dresden          & leg fracture    & Sunday, 16 \\
1             & Miller        & Pirna            & leg fracture    & Sunday, 16 \\ \hline
\end{tabular}
\end{center}
Evidently, the snapshot does not satisfy any primary key since each column features some null marker occurrence, or a duplication of some value.
\end{example}

In response, several researchers proposed the notion of a key set. As the term suggests, a key set is a set of attribute subsets. Naturally, we call the elements of a key set a key. A relation satisfies a given key set if for every pair of distinct rows in the relation there is some key in the key set on which both rows have no null marker occurrences and non-matching values on some attribute of the key. The formal definition of a key set will be given in Definition~\ref{d:key-set} in Section~\ref{s:notation}. The flexibility of a key set over a primary key can easily be recognized, as a primary key would be equivalent to a singleton key set, with the only element being the primary key. Indeed, with a key set different pairs of rows in a relation may be distinguishable by different keys of the key set, while all pairs of rows in a relation can only be distinguishable by the same primary key. We illustrate the notion of a key set on our running example.

\begin{example}\label{ex:intro-2} The relation in Example~\ref{ex:intro} satisfies no primary key. Nevertheless, the relation satisfies several key sets. For example, the key set $\{\{\textit{room}\},\{\textit{time}\}\}$ is satisfied, but not the key set $\{\{\textit{room},\textit{time}\}\}$. The relation also satisfies the key sets \[\mathcal{X}_1=\{ \{\textit{room},\textit{time}\}, \{\textit{injury},\textit{time}\} \}\mbox{ and }\mathcal{X}_2=\{ \{\textit{name},\textit{time}\}, \{\textit{injury},\textit{time}\} \},\] as well as the key set $\mathcal{X}=\{ \{\textit{room},\textit{name},\textit{time}\}, \{\textit{injury},\textit{time}\} \}$.
\end{example}

It is important to point out a desirable feature that primary keys and key sets share. Both are independent of the interpretation of null marker occurrences. That is, any given primary key and any given key set is either satisfied or not, independently of what information any of the null marker occurrences represent. Primary keys and key sets are only dependent on actual values that occur in the relevant columns. This is achieved by stipulating the completeness criterion. The importance of this independence is particularly appealing in modern applications where data is integrated from various sources, and different interpretations may be associated with different occurrences of null markers.

Given the flexibility of key sets over primary keys, and given their independence of null marker interpretations, it seems natural to further investigate the notion of a key set. Somewhat surprisingly, however, neither the research community nor any system implementations have analyzed key sets since their original proposal in 1989. The main goal of this article is to take first steps into the investigation of computational problems associated with key sets. In database practice, one of the most fundamental problems is the implication problem. The problem is to decide whether for a given set $\Sigma\cup\{\varphi\}$ of key sets, every relation that satisfies all key sets in $\Sigma$ also satisfies $\varphi$. Reasoning about the implication of any form of database constraints is important because efficient solutions to the problem enable us to facilitate the processing of database queries and updates.

\begin{example} Recall the key sets $\mathcal{X}_1$, $\mathcal{X}_2$, and $\mathcal{X}$ from Example~\ref{ex:intro-2}. An instance of the implication problem is whether $\Sigma=\{\mathcal{X}_1,\mathcal{X}_2\}$ implies the key set $\varphi=\mathcal{X}$, and another instance is whether $\Sigma$ implies \[\varphi'=\{\{\textit{room}\},\{\textit{name}\}, \{\textit{address}\},\{\textit{time}\}\}.\]
\end{example}

\noindent
\textbf{Contributions.} Our contributions can be summarized as follows.
\begin{itemize}
\item We compare the notion of a key set with other notions of keys. In particular, primary keys are key sets with just one element, and certain keys are unary key sets, for which every key is a singleton.
\item We develop a naive quadratic as well as a linear-time algorithm for validating whether a given key set holds on a given data set, and experimentally demonstrate its efficiency on real-world data sets. The experiments also confirm how additional keys in a given key set help separate tuple pairs in real-world benchmark data.
\item We illustrate how automated reasoning tools for key sets can facilitate efficient updates and queries in database systems.
\item We establish a binary axiomatization for the implication problem of key sets. Here, binary refers to the maximum number of premises that any inference rule in our axiomatization can have. This is interesting as all previous notions of keys enjoy unary axiomatizations, in particular primary keys. What that means semantically is that every given key set that is implied by a set of key sets is actually implied by at most two of the key sets.
\item We establish that the implication problem for key sets is \emph{coNP}-complete. Again, this complexity is quite surprising in comparison with the linear time decidability of other notions of keys.
\item An interesting notion in database theory is that of Armstrong databases. A given class of constraints, such as keys, key sets, or other data dependencies \cite{thalheim:1991}, is said to enjoy Armstrong databases whenever for every given set of constraints in this class there is a single database with the property that for every constraint in the class, the database satisfies this constraint if and only if the constraint is implied by the given set of constraints. This is a powerful property as multiple instances over the implication problem reduce to validating satisfaction over the same Armstrong database. Consequently, the generation of Armstrong databases would create `perfect models' of a given constraint set, which has applications in the acquisition of requirements in database practice. We show that key sets do not enjoy Armstrong relations, as opposed to other classes of keys known from the literature.
\item Our proof techniques for our axiomatizations help us characterize the implication problem of key sets by a fragment of propositional logic under three-valued interpretations. The transformation between key sets and their propositional formulae becomes simpler in the presence of \texttt{NOT NULL} constraints, which disallow any occurrences of missing values in columns for which they are specified in practical database management systems, such as SQL.
\item We then identify an expressive fragment of key sets for which the associated implication problem can be characterized by a unary axiomatization and a quadratic-time algorithm. The fragment also enjoys Armstrong relations and we show how to generate them with conservative use of time and space.
\end{itemize}

From a conceptual point of view, previous research has demonstrated the elegance and robustness of using key sets to control entity integrity in databases with missing values. Our article establishes limitations and opportunities for the use of keys sets in controlling entity integrity from a computational point of view. These include their validation, minimization of overheads for integrity control, and computational support for the acquisition of key sets that encode application semantics. Figure~\ref{fig:overview} illustrates the areas of our contributions and the sequence of the associated computational problems we address in this article.

\begin{figure}
\centering\includegraphics[width=0.75\textwidth]{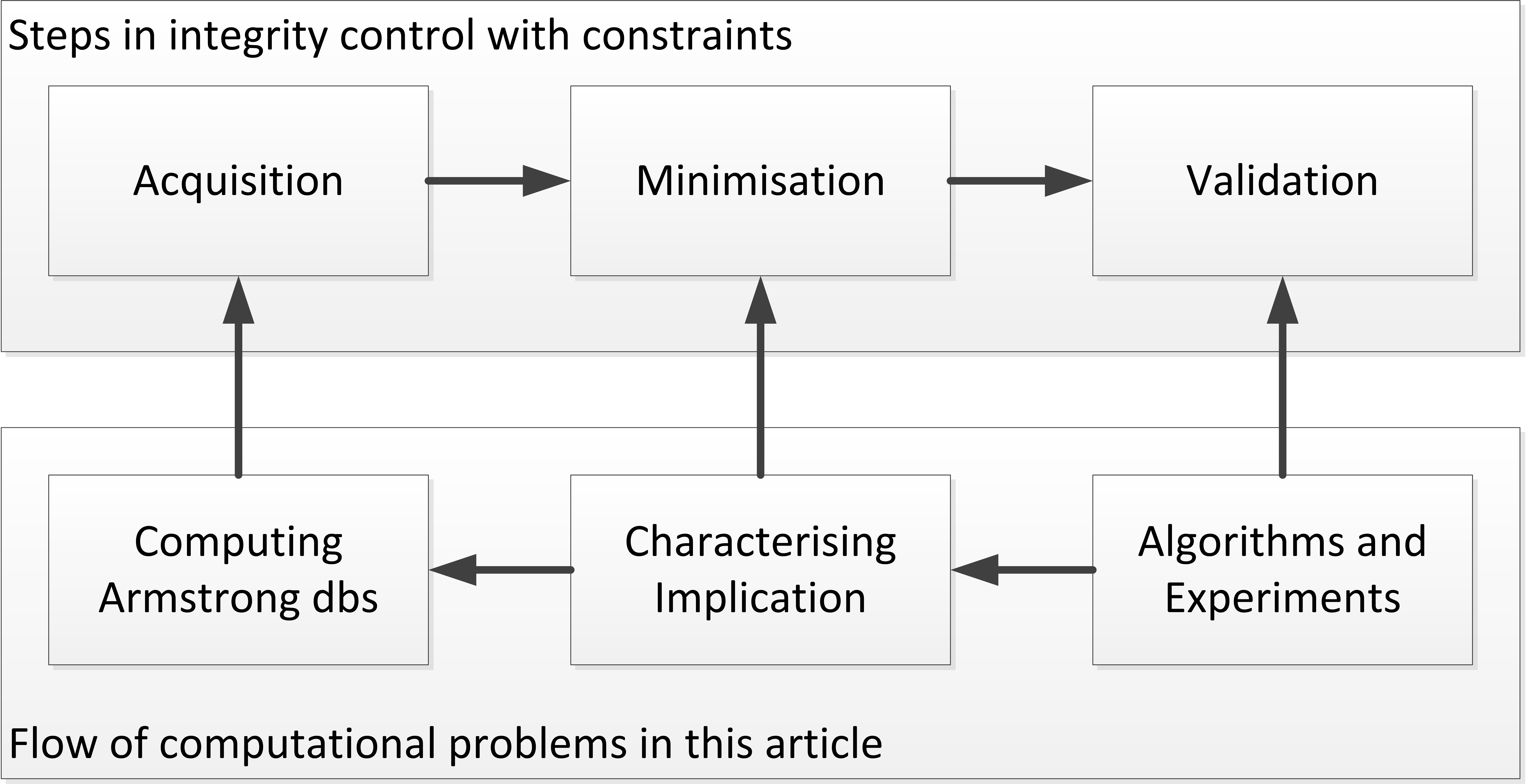}
\caption{Entity Integrity Control with Key Sets}\label{fig:overview}
\end{figure}

\noindent
\textbf{Organization.} We discuss related work in Section~\ref{s:related}. Basic notions and notation are fixed in Section~\ref{s:notation}. The validation problem is addressed in Section~\ref{s:validation}. Section~\ref{s:applications} discusses applications of key sets in the processing of queries and updates. An axiomatization for key sets is established in Section~\ref{s:axioms}. The \textit{coNP}-completeness of the implication problem is settled in Section~\ref{s:complexity}.  A characterization for the implication problem of key sets in terms of a fragment of a three-valued propositional logic is established in Section~\ref{s:logical}. The general existence of Armstrong relations is dis-proven in Section~\ref{s:Armstrong}. A computationally friendly fragment of key sets is identified in Section~\ref{s:fragment}. We conclude and briefly discuss future work in Section~\ref{s:conclusion}.

\section{Related Work}\label{s:related}

We provide a concise discussion on the relationship of key sets with other notions of keys over relations with missing information.

Codd is the inventor of the relational model of data \cite{Codd:1970:RMD:362384.362685}. He proposed the rule of entity integrity, which stipulates that every entity in every table should be uniquely identifiable. In SQL that led to the introduction of primary keys, which stipulate uniqueness and completeness on the attributes that form the primary key. The primary key is a distinguished candidate key. We call an attribute set a \emph{candidate key} for a given relation if and only if every pair of distinct tuples in the relation has no null marker occurrences on any of the attributes of the candidate key and there is some attribute of the candidate key on which the two tuples have non-matching values. The notions of primary and candidate keys have been introduced very early in the history of database research \cite{DBLP:journals/jcss/LucchesiO78}. Candidate keys are singleton key sets, that is, key sets with just one element (namely the candidate key). Hence, instead of having to be complete and unique on the same combination of columns in a candidate key, key sets offer different alternatives of being complete and unique for different pairs of tuples in a relation. Candidate keys were studied in \cite{DBLP:journals/cj/HartmannLL11}. In that work, the associated implication problem was characterized axiomatically and algorithmically, the automatic generation of Armstrong relations was established, and extremal problems associated with families of candidate keys were investigated. As Example~\ref{ex:intro} shows, there are relations on which no candidate key holds, but which satisfy key sets.

Lucchesi and Osborn studied computational problems associated with candidate keys \cite{DBLP:journals/jcss/LucchesiO78}. However, their focus was an algorithm that finds all candidate keys implied by a given set of functional dependencies. They also proved that deciding whether a given relation satisfies some key of cardinality not greater than some given positive integer is NP-complete. Recently, this problem was shown to be W[2]-complete in the size of the key \cite{Blasius0S16}. The discovery which key sets hold on a given relation is beyond the scope of this paper and left as an open problem for future work.

Key sets were introduced by Thalheim \cite{DBLP:journals/eik/Thalheim89} as a generalization of Codd's rule for entity integrity. He studied combinatorial problems associated with unary key sets, such as the maximum cardinality that non-redundant families of unary key sets can have, and which families attain them \cite{thalheim:1991,DBLP:journals/dam/Thalheim92}. Key sets were further discussed by Levene/Loizou \cite{DBLP:journals/ita/LeveneL01} where they also generalized Codd's rule for referential integrity. Somewhat surprisingly, the study of the implication problem for key sets has not been addressed by previous work. This is also true for other automated tasks which require reasoning about key sets.

More recently, the notions of possible and certain keys were proposed \cite{DBLP:journals/vldb/KohlerLLZ16}. These notions are defined for relations in which null marker occurrences are interpreted as `no information', and possible worlds of an incomplete relation are obtained by independently replacing null marker occurrences by actual domain values (or the N/A marker indicating that the value does not exist). A key is said to be \emph{possible} for an incomplete relation if and only if there is some possible world of the incomplete relation on which the key holds. A key is said to be \emph{certain} for an incomplete relation if and only if the key holds on every possible world of the incomplete relation. For example, the relation in Example~\ref{ex:intro} satisfies the possible key $p\langle room, name, address\rangle$, since the key \{\emph{room},\emph{name},\emph{address}\} holds on the possible world:
\begin{center}
\begin{tabular}{c@{\hspace*{.5cm}}c@{\hspace*{.5cm}}c@{\hspace*{.5cm}}c@{\hspace*{.5cm}}c}\hline
\textit{room} & \textit{name} & \textit{address} & \textit{injury} & \textit{time} \\ \hline
1             & Miller        & Dresden          & cardiac infarct & Sunday, 19 \\
2             & Maier         & Pirna            & skull fracture  & Monday, 19 \\
2             & Maier         & Dresden          & leg fracture    & Sunday, 16 \\
1             & Miller        & Pirna            & leg fracture    & Sunday, 16 \\ \hline
\end{tabular}
\end{center}
of the relation. In contrast, the key \{\emph{room},\emph{name}\} is not possible for the relation because the first and last tuple will have matching values on room and name in every possible world of the relation. The key \{\emph{address}\} is possible, but not certain, and the key \{\emph{room},\emph{time}\} is certain for the given relation. Now, it is not difficult to see that an incomplete relation satisfies the certain key $c\langle A_1,\ldots,A_n\rangle$ if and only if the relation satisfies the key set $\{\{A_1\},\ldots,\{A_n\}\}$. In this sense, certain keys correspond to key sets which have only singleton keys as elements. The papers \cite{DBLP:journals/vldb/KohlerLLZ16,DBLP:journals/pvldb/KohlerLZ15} investigate computational problems for possible and certain keys with \texttt{NOT NULL} constraints. In the current paper we investigate a different class of key constraints, namely key sets. In particular, the computationally-friendly fragment of key sets we identify in Section~\ref{s:fragment} subsumes the class of certain keys as the special case of unary key sets.

Recently, \emph{contextual keys} were introduced as a means to separate completeness from uniqueness requirements \cite{DBLP:conf/er/WeiLL17}. A contextual key is an expression $(C,X)$ where $X\subseteq C$. These are different from key sets since $X\subseteq C$ is a key for only those tuples that are complete on $C$. In particular, the special case where $C=X$ only requires uniqueness on $X$ for those tuples that are complete on $X$. This captures the \texttt{UNIQUE} constraint of SQL. Indeed, $\texttt{UNIQUE}(A_1,\ldots,A_n)$ holds on a given relation if and only if the possible key $p\langle A_1,\ldots,A_n\rangle$ holds on the relation \cite{DBLP:journals/vldb/KohlerLLZ16}. We leave it as future work to combine key sets and contextual keys into a unifying notion of contextual key sets.

\section{Preliminary Definitions}\label{s:notation}

In this section, we give some basic definitions and fix notation.

A \emph{relation schema} is a finite non-empty set of attributes, usually denoted by $R$. A \emph{relation} $r$ over $R$ consists of tuples $t$ that map each $A\in R$ to $\Dom(A)\cup \{\bot\}$ where $\Dom(A)$ is the domain associated with attribute $A$ and $\bot$ is the unique null marker. Given a subset $X$ of $R$, we say that a tuple $t$ is \emph{$X$-total} if $t(A)\neq \bot$ for all $A\in X$. Informally, a relation schema represents the column names of database tables, while each tuple represents a row of the table, so a relation forms a database instance. Moreover, $\Dom(A)$ represents the possible values that can occur in column $A$ of a table, and $\bot$ represents missing information. That is, if $t(A)=\bot$, then there is no information about the value $t(A)$ of tuple $t$ on attribute $A$.

In our running example, we have the relation schema
\begin{center}
\textsc{Ward}=\{\textit{room},\textit{name},\textit{address},\textit{injury},\textit{time}\}.
\end{center}
Each of these attributes comes with a domain, which we do not specify any further here. Each row of the table in Example~\ref{ex:intro} represents a tuple. The second row, for example, is $\{\textit{injury},\textit{time}\}$-total, but not total on any proper superset of $\{\textit{injury},\textit{time}\}$. The four tuples together constitute a relation over \textsc{Ward}.

The following definition introduces the central object of our studies. It was first defined by Thalheim in \cite{DBLP:journals/eik/Thalheim89}.

\begin{definition}\label{d:key-set}
A \emph{key set} is a finite, non-empty collection $\calX$ of subsets of a given relation schema $R$. We say that a relation $r$ over $R$ \emph{satisfies} the key set $\calX$ if and only if for all distinct $t,t'\in r$ there is some $X\in \calX$ such that $t$ and $t'$ are $X$-total and $t(X)\neq t'(X)$. Each element of a key set is called a \emph{key}. If all keys of a key set are singletons, we speak of a \emph{unary key set}.
\end{definition}

In the sequel we write $\calX,\calY,\calZ,\ldots $ for key sets and $X,Y,Z,\ldots$ for attribute sets, and $A,B,C,\ldots $ for attributes. We sometimes write $A$ instead of $\{A\}$ to denote the singleton set consisting of only $A$. If $\tuple X$ is a sequence, then we may sometimes write simply  $\tuple X$ for the set that consists of all members of $\tuple X$. In the following discussion, we use $|K|$ to denote the number of keys in the key set $K$, and $\|K\|$ to denote the total number of attribute occurrences in $K$.

As already mentioned in Example~\ref{ex:intro-2}, the relation in Example~\ref{ex:intro} satisfies the key sets $\mathcal{X}_1$, $\mathcal{X}_2$, and $\mathcal{X}$. It also satisfies the unary key set $\{\{\emph{room}\},\{\emph{time}\}\}$, but not the singleton key set $\{\{\emph{room},\emph{time}\}\}$. In Example~\ref{ex:intro-2}, we have $|K|=2$ and $\|K\|=5$.

A fundamental problem in automated reasoning about any class of constraints is the \emph{implication problem}. For key sets, the problem is to decide whether for an arbitrary relation schema $R$, and an arbitrary set $\Sigma\cup\{\varphi\}$ of key sets over $R$, $\Sigma$ implies $\varphi$. Indeed, $\Sigma$ \emph{implies} $\varphi$ if and only if every relation over $R$ that satisfies all key sets in $\Sigma$ also satisfies the key set $\varphi$. The following section illustrates how solutions to the implication problem of key sets can facilitate the efficient processing of queries and updates.

\section{Validating Key Sets}\label{s:validation}

In this section we will investigate the validation problem for key sets. This problem takes as input both a key set $\{X_1,\ldots,X_n\}$ and a given relation $r$, and returns `yes' if $\{X_1,\ldots,X_n\}$ satisfies $r$, and returns `no' otherwise. The validation problem is one of the most basic decision problems that are fundamental for automating integrity control management. Efficient solutions to this problem enable us to validate whether the given data - possible after updates - is compliant with the business rules encoded by the key set. In other words, a computer can check quickly whether the data is compliant or not. Throughout this section, we will use the following example.

\begin{exmp}\label{example:keys}
	Table~\ref{Incomplete relation} shows an incomplete relation $r$, where $\perp$ denotes a null marker occurrence. This relation does not satisfy any candidate key as null markers occur in columns \emph{address} and \emph{injury}, and tuples $t_1$ and $t_2$ have the same projection on $\{\emph{name},\emph{time}\}$. The relation only satisfies one certain key \cite{DBLP:journals/vldb/KohlerLLZ16}, which is $\{\emph{name},\emph{address},\emph{injury},\emph{time}\}$ as the two null marker occurrences can be replaced by any domain values without causing any duplication. For the key set $\cal X=\{\{\emph{name},\emph{address}\},\{\emph{injury}\},\{\emph{time}\}\}$ we obtain the following: Tuple pairs $t_1$ and $t_3$, and $t_1$ and $t_4$ are complete and different on $X=\{\emph{name},\emph{address}\}$, tuple pairs $t_1$ and $t_2$, and $t_2$ and $t_3$ are complete and different on $Y=\{\emph{injury}\}$, and tuple pairs $t_2$ and $t_4$, and $t_3$ and $t_4$ are complete and different on $Z=\{\emph{time}\}$. Hence, for every pair of different tuples there is some $X\in\cal X$ such that the tuple pair is complete and different on $X$. Hence, the relation $r$ satisfies the key set $X$.
\end{exmp}

\begin{table}[H]
	\centering
	\caption{A relation with missing data}
	\label{Incomplete relation}
\begin{tabular}{c@{\hspace*{.25cm}}|@{\hspace*{.25cm}}c@{\hspace*{.5cm}}c@{\hspace*{.5cm}}c@{\hspace*{.5cm}}c@{\hspace*{.5cm}}c}\hline
      & \textit{name} & \textit{address}  & \textit{injury} & \textit{time} \\ \hline
$t_1$ & Miller        & Dresden          & cardiac infarct & Sunday, 19 \\
$t_2$ & Miller        & $\perp$          & skull fracture  & Sunday, 19 \\
$t_3$ & Maier         & Dresden          & cardiac infarct & Sunday, 19 \\
$t_4$ & Maier         & Dresden          & $\perp$         & Monday, 20 \\ \hline
\end{tabular}
\end{table}

\subsection{Na\"{\i}ve validation}

The validation problem is a decision problem, so the output to every input is either `yes' or `no'. For instances that result in a `no', the answer is rather uninformative and the first reaction is probably to ask why. So, rather than looking at the validation problem, it makes more sense to return the set $r^{\cal X}_V\subseteq r$ of tuples in $r$ that violate a key set $\cal X$. This set $r^{\cal X}_V$ is the maximum subset of $r$ under subset inclusion with the property that if $t\in r^{\cal X}_V$, then there is some $t'\in r^{\cal X}_V$ with $t\not=t'$, and for all $X\in \cal X$, $t$ is not $X$-total or $t'$ is not $X$-total or $t(X)=t'(X)$. If $r^{\cal X}_V$ is empty, then the input relation satisfies the key set $\cal X$.

\begin{algorithm}[H]
	\caption{Quadratic-time algorithm\label{alg:naive}}
	\begin{algorithmic}
		\Require A key set $\cal X$, a relation $r$ over relation schema $R$
		\Ensure The maximum set $r^{\cal X}_V$ of tuples in $r$ that violate the key set $\cal X$
		\State $r^{\cal X}_V:=\emptyset$
		\For{$t,t'\in r$ where $t\neq t'$}
		\For{$X\in\cal X$}
		\If{for all $X\in\cal X$, ($t$ or $t'$ is not $X$-total) or $(t(X)=t'(X))$}
		\State $r^{\cal X}_V:=r^{\cal X}_V\cup\{t\}$
		\EndIf
		\EndFor
		\EndFor
		\Return $r^{\cal X}_V$
	\end{algorithmic}
\end{algorithm}

We begin with a na\"{\i}ve algorithm that computes the set $r^{\cal X}_V$ given relation $r$ and given key set $\cal X$ in time quadratic in the input. The pseudo-code is given in Algorithm~\ref{alg:naive}. The algorithm starts with an empty set $r^{\cal X}_V$. The first loop selects every pair of distinct tuples in the input relation $r$. Then, the algorithm will go through the given key set $\cal X$ and check if the pair of selected tuples violate all keys $X$ in the key set $\cal X$. If that is the case, both tuples will become part of $r^{\cal X}_V$. After every pair of distinct tuples in the input relation has been checked, the algorithm returns $r^{\cal X}_V$. It is evident that the algorithm is correct as it strictly follows the definition of a key set. The time complexity of Algorithm~\ref{alg:naive} is $\mathcal{O}(|r|^2\cdot \|\cal X\|)$ since every pair of distinct tuples is evaluated against every attribute occurrence in $\cal X$.

\begin{theorem}
Given a relation $r$ and a key set $\cal X$, Algorithm~\ref{alg:naive} computes the set $r^{\cal X}_V$ of violating tuples in time $\mathcal{O}(|r|^2\cdot \|\cal X\|)$.\qed
\end{theorem}

\subsection{Linear-time validation}

In Algorithm~\ref{alg:linear}, our strategy is to aggressively partition the input relation $r$ into different smaller subsets $b\subseteq r$. Strictly speaking, the subsets $b$ do not form a partition as i) incomplete tuples may occur in multiple subsets, and ii) tuples with unique complete projections on some key do not need to be tracked since they cannot contribute to the violation of the given key set. Indeed, each subset $b$ contains tuples which are either incomplete or have matching values on all previously examined keys of the input key set. By examining each key $X$ in the input key set $\mathcal{X}$, the subsets $b$ are progressively split into smaller subsets. Intuitively, each subset $b$ contains tuples such that every pair of distinct tuples from $b$ violate all the keys in the input key set that have been examined so far. The worst-case time complexity for Algorithm~\ref{alg:linear} is $O(|r|\cdot\|K\|)$. This is a huge improvement over Algorithm~\ref{alg:naive}. Before we present Algorithm~\ref{alg:linear} formally, we show how Algorithm~\ref{alg:linear} works on a real example.

\begin{exmp}
	Let $R=\{\emph{n(ame)},\emph{a(ddress)},\emph{i(njury)},\emph{t(ime)}\}$, $r = \{t_1,t_2,t_3,t_4\}$ and the key set $\mathcal{X} = \{\{n,a\},\{i\},\{t\}\}$. The output will be set $B_{\cal X}$ of subsets $b$ of $r$ such that for every $b\in B_{\cal X}$ and for all distinct $t,t'\in b$, $\{t,t'\}$ violate the key set $\mathcal{X}$. If no such $b$ exists, then the input relation satisfies the input key set.
	
Algorithm~\ref{alg:linear} starts with $B_{\cal X}=\{r\}$ where the only block is the input relation $r$.
\begin{itemize}
\item In the first iteration we pick $X=\{n,a\}\in\mathcal{X}$, and set $B:=\emptyset$. For our map $M$ and set $I$ of incomplete tuples we obtain:
    \[M(\text{Miller},\text{Dresden})=\{t_1\}, M(\text{Maier},\text{Dresden})=\{t_3,t_4\}\mbox{, and }I=\{t_2\}\]
    We then add every tuple in $I$ into each image where $M$ is defined, so we obtain the following for our map:
    \[M(\text{Miller},\text{Dresden})=\{t_1,t_2\}, M(\text{Maier},\text{Dresden})=\{t_2,t_3,t_4\}\]
    As $B$ is the set of all images of values on which $M$ is defined we have:
    \[B_{\cal X}=B=\{\{t_1,t_2\},\{t_2,t_3,t_4\}\}\]
    at the end of the first iteration.

\item In the second iteration we pick $X=\{i\}\in\mathcal{X}$, and set $B:=\emptyset$. For $b=\{t_1,t_2\}$ we obtain:
\[ M(\text{cardiac infarct})=\{t_1\}, M(\text{skull fracture})=\{t_2\}\mbox{, and}I=\emptyset\;.\]
No further changes are applied as we only have singleton images in the map and no incomplete tuples exist, so $B=\emptyset$. For $b=\{t_2,t_3,t_4\}$ we obtain:
\[ M(\text{skull fracture})=\{t_2\}, M(\text{cardiac infarct})=\{t_3\}\mbox{, and }I=\{t_4\}\;.\]
We then add every tuple in $I$ into each image where $M$ is defined, so we obtain:
\[ M(\text{skull fracture})=\{t_2,t_4\}, M(\text{cardiac infarct})=\{t_3,t_4\}\;.\]
Since $B$ is the set of all images of values on which $M$ is defined we have:
\[B_{\cal X}=B=\{\{t_2,t_4\},\{t_3,t_4\}\}\]
at the end of the second iteration.

\item For the third and final iteration we pick $X=\{t\}\in\mathcal{X}$, and set $B:=\emptyset$. For $b=\{t_2,t_4\}$ we obtain:
\[ M(\text{Sunday, 19})=\{t_2\}, M(\text{Monday, 20})=\{t_4\}\mbox{, and }I=\emptyset\;.\]
No further changes are applied as we only have singleton images in the map and no incomplete tuples exist, so $B=\emptyset$. For $b=\{t_3,t_4\}$ we obtain
\[M(\text{Sunday, 19})=\{t_3\}, M(\text{Monday, 20})=\{t_4\}\mbox{, and }I=\emptyset\;.\]
Again, no further changes are applied as we only have singleton images in the map and no incomplete tuples exist, so $B=\emptyset$.
\end{itemize}
As final output we return $B_{\cal X}=\emptyset$. That means $r$ satisfies $\mathcal{X}$.
\end{exmp}

Algorithm~\ref{alg:linear} shows the pseudo-code for computing the set $B_{\cal{X}}$ of all $\subseteq$-maximal subsets  $b\subseteq r$ such that for all pairs of distinct tuples $t,t'\in b$, $\{t,t'\}$ violates $\cal X$.  

\begin{algorithm}
	\caption{Linear-time Algorithm\label{alg:linear}}
	\begin{algorithmic}
		\Require A key set $\cal X$, a relation $r$ over relation schema $R$
		\Ensure Set $B_{\cal X}$ of $\subseteq$-maximal subsets $b\subseteq r$ such that for all distinct tuple pairs $t,t'\in b$, $\{t,t'\}$ violate $\cal X$
		\State $B_{\cal X}:=\{r\}$
        \For{$X\in\cal X$}
		\State $B:=\emptyset$
		\ForAll{$b\in B_{\cal X}$}
		\State Create an empty map $M$%\comment{maps values $y$ to set of tuples $t$ with $t[X]=y$}
		\State Create an empty set $I$%\comment{contains all tuples with missing values on $X$}
		\ForAll{$t\in b$}
		\If{$t[X]$ contains some missing value}
		\State $I:=I\cup\{t\}$
		\Else
		\State $M[t[X]]:=M[t[X]]\cup\{t\}$
		\EndIf
		\EndFor
		\ForAll{$y$ such that $M[y]$ exists}
		\State $M[y]:=M[y]\cup I$%\comment{incomplete tuples always offend}
		\If{$|M[y]|>1$}
		\State $B:=B\cup\{M[y]\}$
		\EndIf
		\EndFor
		\If{there is no $y$ such that $M[y]$ exists and $|I|>1$}
		\State $B:=B\cup\{I\}$
		\EndIf
		\EndFor
	    \State $B_{\cal X}:=B$
		\EndFor
		\Return $B_{\cal X}$
	\end{algorithmic}
\end{algorithm}

Evidently, different tuples will and can only occur in the same subset $b$ at the end of the computation, if it is true for every key in the key set that they have either matching complete values on all the columns of the key, or one of the tuples has a null marker occurrence on some column in the key.  

\begin{theorem}
Given a relation $r$ and a key set $\cal X$, Algorithm~\ref{alg:linear} computes the set $B_\mathcal{X}$ of all maximal subsets $b\subseteq r$ such that for all distinct $t,t'\in b$, $\{t,t'\}$ violates $\mathcal{X}$, in time $\mathcal{O}(|r|\cdot \|\cal X\|)$.\qed
\end{theorem}

\subsection{Experiments}

For the benefit of complementing our theoretical worst-case time complexity analysis with some actual runtimes of our algorithms, we conducted a few experiments with publicly available data sets. These have served as benchmark data sets for the discovery problem of a variety of data dependencies \cite{}.  Our algorithms and experiments were implemented in Python 3.

\noindent
\textbf{Data sets}. Our experiments are based on a collection of real-world data sets with missing data. Basic statistics of the data sets are shown in Table~\ref{datasets}. We use $\#R$,  $\#C$ and $\#\perp$ to denote the number of rows, columns and missing data values in the given data set, respectively. As we can tell, these real-world data sets have indeed different characteristics. Particularly, more than 40\% of values in the data set {\em plista} are missing, but only 7.5\% of values are missing in the data sets {\em hepatitis} and {\em echocardigram}.

\begin{table}[t]
	\centering
	\caption{Incomplete Datasets}
	\label{datasets}
	\begin{tabular}{|l|l|l|l|}
		\hline
		Dataset                 & \#R  & \#C & \#$\perp$     \\ \hline
		horse                   & 300  & 28  & 1605       \\ \hline
		bridges                 & 108  & 13  & 77          \\ \hline
		hepatitis               & 155  & 20  & 167          \\ \hline
		breast-cancer-wisconsin & 691  & 11  & 16           \\ \hline
		echocardiogram          & 132  & 13  & 132         \\ \hline
		plista                  & 996  & 63  & 23317        \\ \hline
		flight                  & 1000 & 109 & 51938       \\ \hline
		ncvoter                 & 1000 & 19  & 2863        \\ \hline
	\end{tabular}
\end{table}

\noindent
\textbf{Generation of key sets}. For our experiments, we use two algorithms to generate synthetic key sets. For each data set with $n$ attributes, the first algorithm creates $n$ key sets. For each key set $\mathcal{X}_i$ where $1 \leq i \leq n$, $$\mathcal{X}_i =\{\{A_1,\ldots, A_i\},\{A_{i+1}\},\ldots, \{A_n\} \}$$ where $A_i$ denotes the $i$-th attribute in the data set. The second algorithm is similar to the first, but for each key set, the first key is randomly selected and the remaining attributes are used to construct singleton keys. For example, if we generate a key set $\mathcal{X} = \{k_1, \ldots, k_{n+1-m}\}$ over a relation schema $R$ with $n$ attributes, then $k_1$ will be generated by randomly selecting $m$ attributes from $R$, and $k_2,\ldots,k_{n+1-m}$ will each denote a singleton subset made up of the remaining $n-m$ attributes of $R$.

\subsection{Experiment 1 - Run time efficiency}

\begin{example}
	In this experiment, we use the first key set generation algorithm with cardinalities from 1 to $n$, where $n$ denotes the number of attributes in a data set. We learn how $|\mathcal{X}|$ affects the satisfaction of the key set $\mathcal{X}$. For example, in the data set 'bridges', a decrease of $|K|$ from 11 to 10 results in a violation of the key set \[\mathcal{X}'=\{\{0,1,2\},\{3\},\{4\},\{5\},\{6\},\{7\},\{8\},\{9\},\{10\},\{11\},\{12\}\}\;.\] Indeed, with the key set \[\mathcal{X}=\{\{0,1,2,3\},\{4\},\{5\},\{6\},\{7\},\{8\},\{9\},\{10\},\{11\},\{12\}\},\] the following six tuples contribute to a violation: 
	\begin{center}
	('E54', 'Y', '?', '1908', 'HIGHWAY', '1240', '?', 'G', '?', 'STEEL', 'MEDIUM', 'F', 'SIMPLE-T'), \\
	('E100', 'O', '43', '1982', 'HIGHWAY', '?', '?', 'G', '?', '?', '?', 'F', '?'),  \\
	('E56', 'M', '23', '1909', 'HIGHWAY', '?', '?', 'G', 'THROUGH', 'STEEL', 'MEDIUM', 'F', 'SIMPLE-T'), \\
	('E40', 'M', '22', '1893', 'HIGHWAY', '?', '2', 'G', 'THROUGH', 'STEEL', 'MEDIUM', 'F', 'SIMPLE-T'),  \\
	('E109', 'A', '28', '1986', 'HIGHWAY', '?', '?', 'G', '?', '?', '?', 'F', '?'), and \\
	('E39', 'A', '25', '1892', 'HIGHWAY', '?', '2', 'G', 'THROUGH', 'STEEL', 'MEDIUM', 'F', 'SIMPLE-T')\;.
   \end{center}
	The main reason of this violation is due to the null marker occurrence on attribute 2 in tuple E54.  For example, tuples $E54$ and $E56$ violate the key set $\mathcal{X}$ because: i) the null marker occurrence on attribute 2 in tuple $E54$ violates the key $\{0,1,2,3\}$, and ii) on all the singleton keys in this key set, both tuples have matching values or a null marker occurrence (and thereby violating each of these keys as well). 
\end{example}

In our first experiment we determine the run-time of our two algorithms for validating key sets on our real-world data sets. For each data set, we use the synthetic key sets generated by the first key set generation algorithm with cardinalities from 1 to 
$n$, where $n$ denotes the number of attributes in a data set. For each key set, we run each algorithm 10 times. We report the average run-time of our two algorithms for each of our data sets, where the average is taken over all runs over all cardinalities. The results are shown in Table~\ref{datasets2}.

\begin{table}[H]
	\centering
	\caption{Run-time comparison (in s) for validating synthetic key sets on benchmark data}
	\label{datasets2}
	\begin{tabular}{|l|l|l|}
		\hline
		Dataset                  & Algorithm~\ref{alg:naive}    & Algorithm~\ref{alg:linear}   \\ \hline
		horse                    & 62.628  & 2.228 \\ \hline
		bridges                  & 1.991   & 0.063   \\ \hline
		hepatitis                & 8.361   & 0.443   \\ \hline
		breast-cancer-wisconsin  & 64.266  & 0.104   \\ \hline
		echocardiogram           & 2.822   & 0.290   \\ \hline
		plista                   & too long     & 2.379   \\ \hline
		flight                   & too long     & 1.958  \\ \hline
		ncvoter                  & 370.416 & 8.036   \\ \hline
	\end{tabular}
\end{table}

The results show clear practical run-time benefits for developing a sophisticated linear-time algorithm that validates key sets. 

\subsection{Experiment 2 - Run time efficiency by cardinality of key sets}

In this experiment, we break down our analysis of the run-time efficiencies of Algorithm~\ref{alg:naive} and Algorithm~\ref{alg:linear} based on a given cardinality of the key sets. We use the second key set generation algorithm to generate random key sets. For each given cardinality that applies to a given data set, we create 100 different key sets randomly. Then we apply our algorithms to each of the key sets. This enables us to observe the run-time of these algorithms in relationship to the cardinalities of a given key set. Interestingly, the results we obtain for the two algorithms are quite different. The run-time behaviors of our two algorithms are visualized in Figure~\ref{exp2-alg1} and Figure~\ref{exp2-alg2}, respectively.

%Recall our second key set generation algorithm, the number of attributes in the random key is decreasing while $|K|$ increasing each time.

%The result of Algorithm 2 is normally distributed overall. The execution time is first increasing until $|K|$ becomes nearly half of $\|K\|$, the execution time starts to decrease.

\begin{figure}[H]
	\begin{multicols}{2}
		\includegraphics[width=\linewidth]{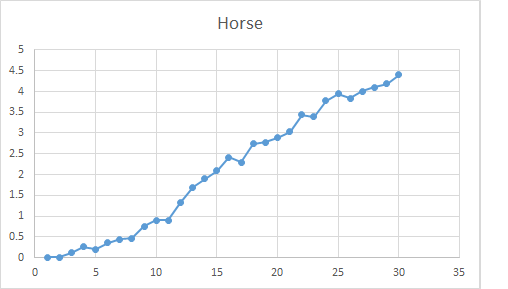}\par
		\includegraphics[width=\linewidth]{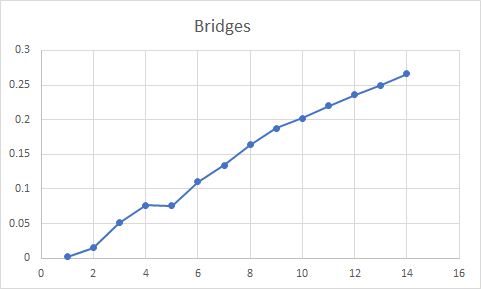}\par
	\end{multicols}
	\begin{multicols}{2} 
		\includegraphics[width=\linewidth]{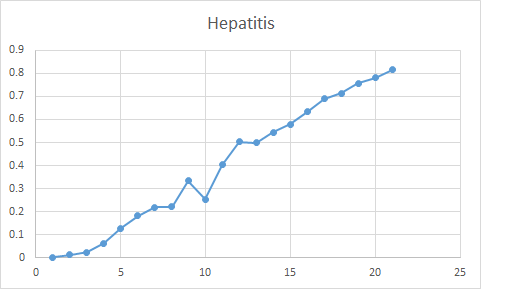}\par
		\includegraphics[width=\linewidth]{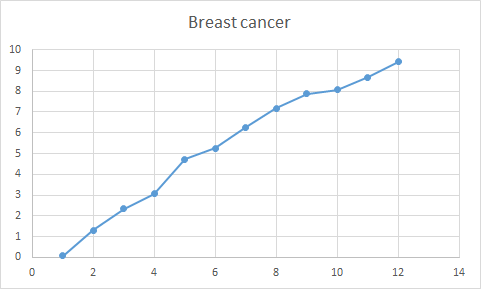}\par
	\end{multicols}
    \begin{multicols}{2}
    	\includegraphics[width=\linewidth]{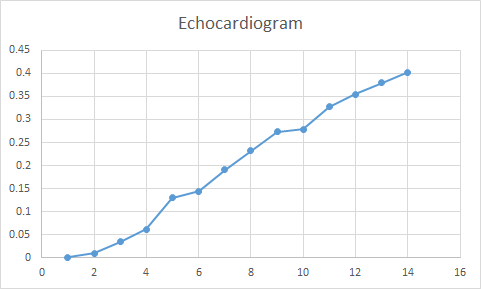}\par
    \includegraphics[width=\linewidth]{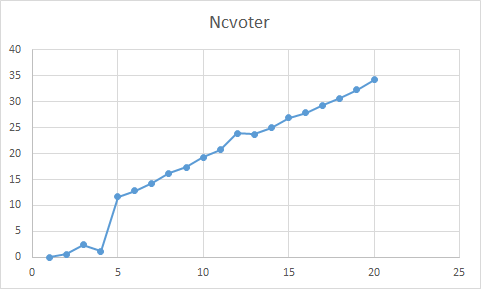}\par	
    \end{multicols}	
	\caption{Run-time of Algorithm~\ref{alg:naive} by the cardinality of the given key sets\label{exp2-alg1}}
\end{figure}

From Figure~\ref{exp2-alg1} we can observe a linear increase of the run-time in the cardinality of the key sets, where the size of the key set is fixed. Indeed, Algorithm~\ref{alg:naive} analyses all the given tuple pairs for each of the keys in the given key set.

\begin{figure}[H]
	\begin{multicols}{2}
		\includegraphics[width=\linewidth]{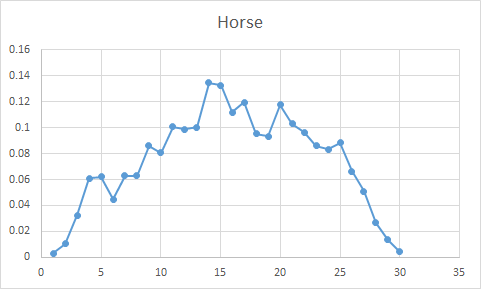}\par
		\includegraphics[width=\linewidth]{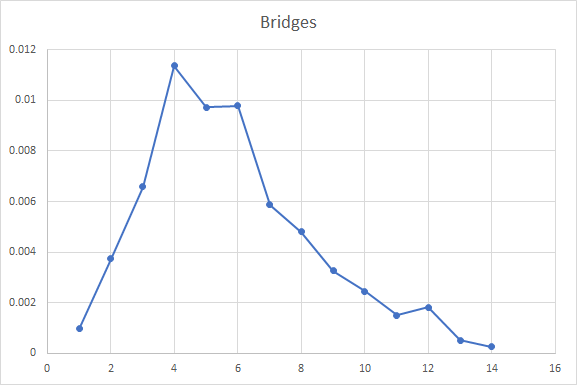}\par
	\end{multicols}
	\begin{multicols}{2}
		\includegraphics[width=\linewidth]{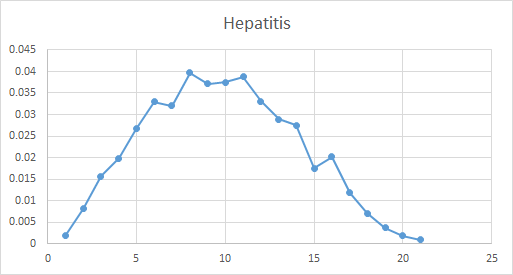}\par
		\includegraphics[width=\linewidth]{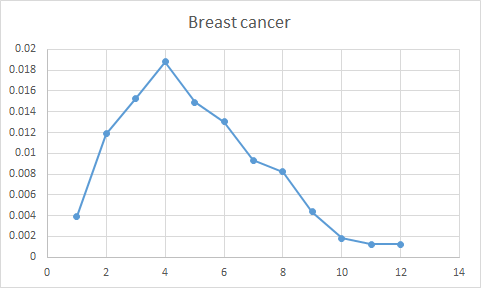}\par
	\end{multicols}
    \begin{multicols}{2}
		\includegraphics[width=\linewidth]{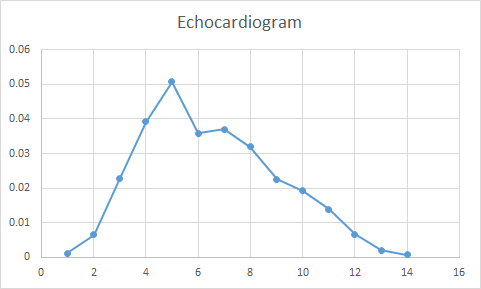}\par
        \includegraphics[width=\linewidth]{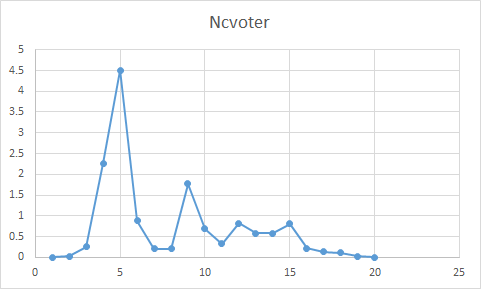}\par		
   \end{multicols}
	\begin{multicols}{2}
		\includegraphics[width=\linewidth]{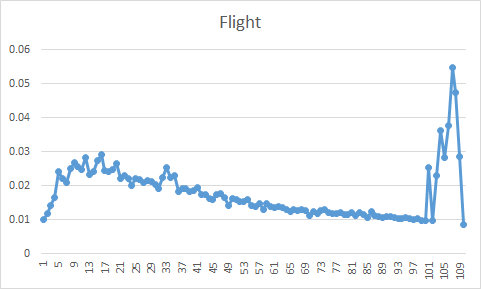}\par	
		\includegraphics[width=\linewidth]{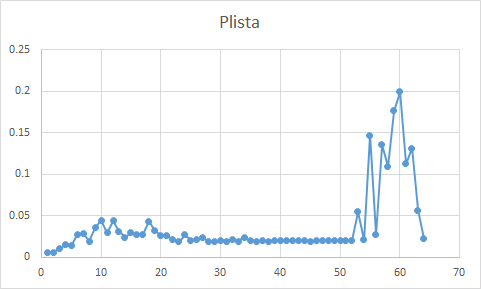}\par
	\end{multicols}
	\caption{Run-time of Algorithm~\ref{alg:linear} by the cardinality of the given key sets\label{exp2-alg2}}
\end{figure}

However, the run-time behavior of Algorithm~\ref{alg:linear} does not follow this pattern. The reason is that the number of tuples that needs to be analyzed by Algorithm~\ref{alg:linear} becomes less as more keys of the given key set are examined. Hence, the total number of tuples examined by Algorithm~\ref{alg:linear} decreases for a larger number of keys in the given key set. In fact, the run-time behavior of Algorithm~\ref{alg:linear} shows actual insight. Indeed, for different key sets of the same size and different cardinalities, there are more tuples that contribute to the violation of key sets that have lower cardinality than other key sets. 

\subsection{Experiment 3 - Offending Tuples}

We will now illustrate experimentally how the increase in the cardinality of the given key set affects the number of offending tuples. For that purpose we have applied Algorithm~\ref{alg:linear}. In Figure~\ref{fig:tuples}, we show the average number of  tuples that offend the given key set found by our previous experiment. The figure shows how the average number of offending tuples decreases with the cardinality of the given key set. Again, this explains why Algorithm~\ref{alg:linear} performs efficiently: the more of the keys of the given key set have been processed, the fewer tuples remain that offend all of those keys.

%Therefore, Algorithm 2 is much more efficient. Notice that the medians of violating tuples found in dataset ``hepatitis" are much lower than their mean when the cardinality of a key set is larger than 4. We also find that the key sets with higher cardinalities result in less number of violating tuples in the real world datasets.
%In order to do further study on the relationship between $|K|$ and violation. We calculate the mean and median of violate tuples for each $|K|$ over a number of runs.

\begin{figure}[H]
	\begin{multicols}{2}
		\includegraphics[width=\linewidth]{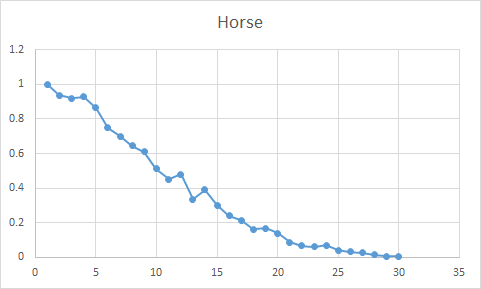}\par
		\includegraphics[width=\linewidth]{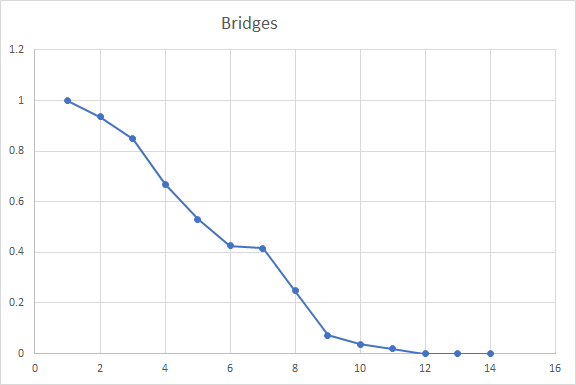}\par
	\end{multicols}
	\begin{multicols}{2}
		\includegraphics[width=\linewidth]{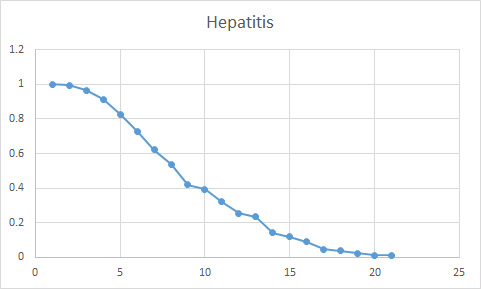}\par
		\includegraphics[width=\linewidth]{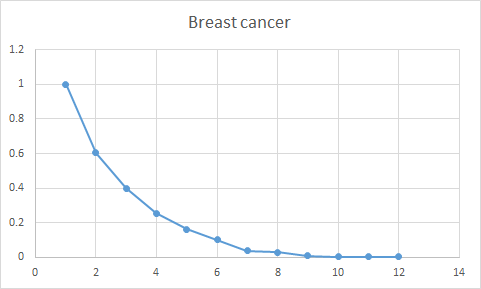}\par
	\end{multicols}
   \begin{multicols}{2}
		\includegraphics[width=\linewidth]{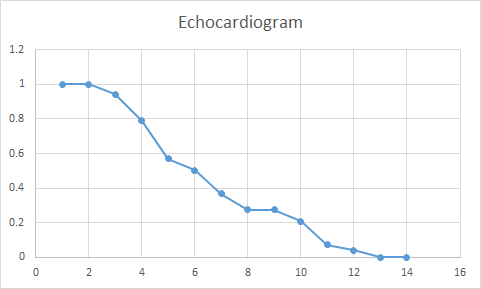}\par
        \includegraphics[width=\linewidth]{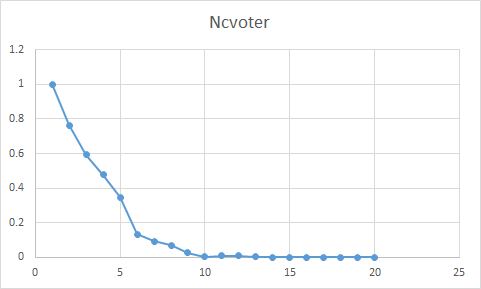}\par
   	\end{multicols}
	\begin{multicols}{2}
		\includegraphics[width=\linewidth]{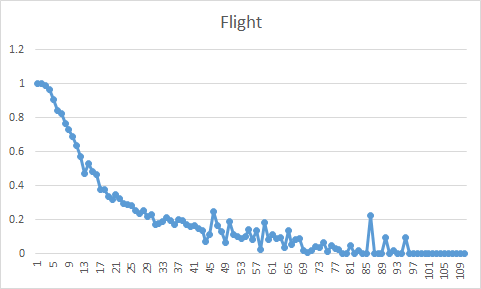}\par	
		\includegraphics[width=\linewidth]{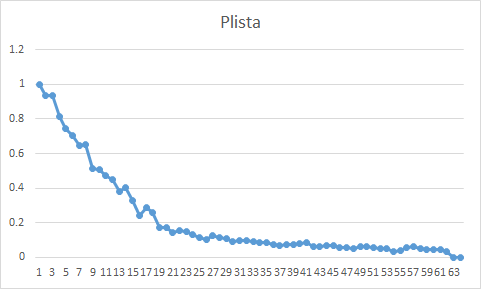}\par
	\end{multicols}
	\caption{Average number  of offending tuples by the cardinality of the given key set\label{fig:tuples}}
\end{figure}

\subsection{Experiment 4 - Adding Keys to Establish Entity Integrity}

For our final experiment, we want to illustrate how the percentage of violations decreases with a growing number of keys in a key set. This confirms our main motivation for studying key sets as a major mechanism to establish entity integrity in relations with missing values. For our experiment, we looked at the benchmark data sets and randomly generated keys sets with different cardinalities to test how easy entity integrity can be established. For a fixed cardinality on a given data set, we randomly generated 100 key sets and recorded the percentage of violations. The results are shown in Figure~\ref{fig:violations}. 

%With same $|K|$, we recorded each violation over total running times. Therefore, we could find the relation between violation time and $|K|$. The result is a negative curvature. With the increase of $|K|$, the violation time over total running time is decreasing.

\begin{figure}[H]
	\begin{multicols}{2}
		\includegraphics[width=\linewidth]{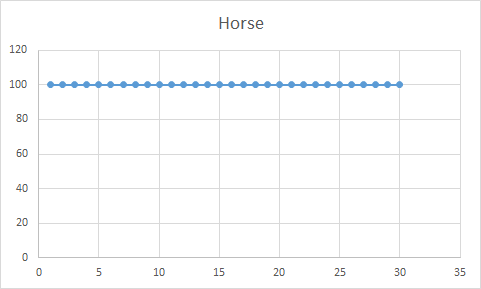}\par
		\includegraphics[width=\linewidth]{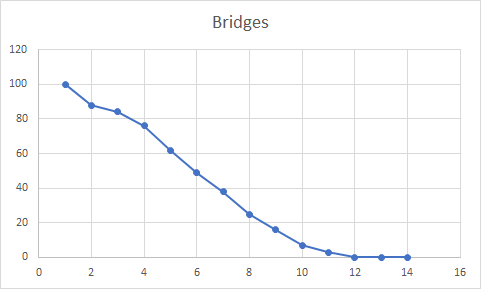}\par
	\end{multicols}
	\begin{multicols}{2}
		\includegraphics[width=\linewidth]{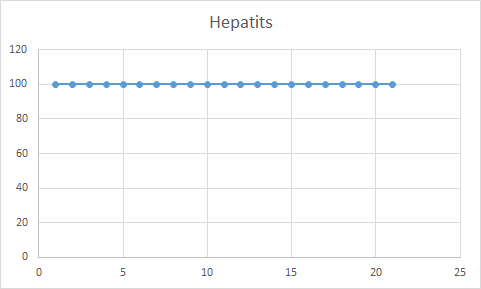}\par
		\includegraphics[width=\linewidth]{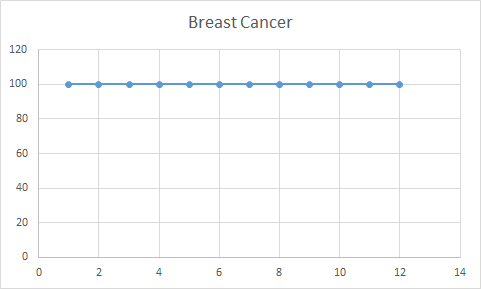}\par
	\end{multicols}
   \begin{multicols}{2}
		\includegraphics[width=\linewidth]{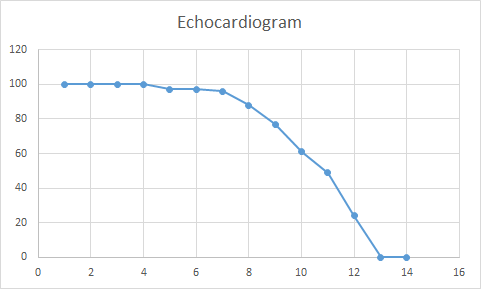}\par
        \includegraphics[width=\linewidth]{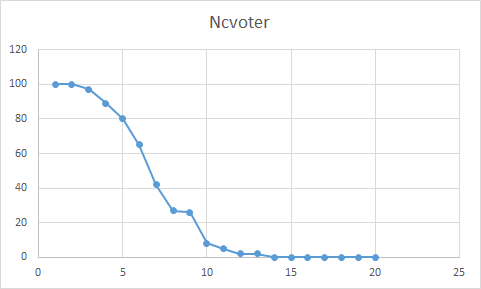}\par		
   \end{multicols}
	\begin{multicols}{2}
		\includegraphics[width=\linewidth]{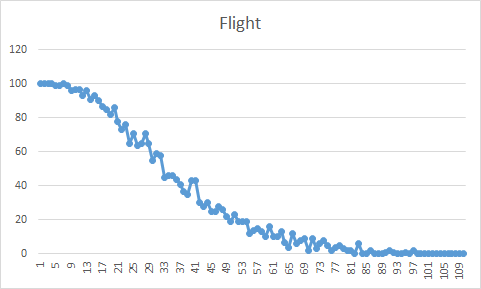}\par	
		\includegraphics[width=\linewidth]{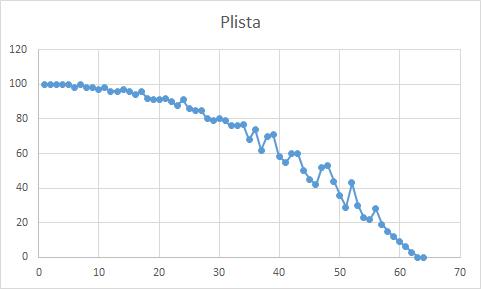}\par
	\end{multicols}
	\caption{Percentage of Violations by Cardinality of Key Sets\label{fig:violations}}
\end{figure}

The data sets ``Horse" , ``Breast Cancer", and ``Hepatitis" illustrate the reality that data sets with duplicate tuples can never achieve entity integrity. That is, whenever there are two different tuples with matching values on all columns, then no key can distinguish them, and therefore also no key set. However, when a data set does not contain any duplicate tuples, then our experiment illustrates clearly how additional keys can further separate tuples that are indistinguishable by previously used keys. 

\subsection{Summary}

Key sets serve as a natural mechanism to establish entity integrity in data sets with missing values. Indeed, the restriction to separate all pairs of distinct tuples by the same key seems unnatural. For example, when we want to identify a person we may want to use biometric measures such as fingerprints or retina scans. Indeed, if one technology fails to obtain the measure, the other one may still work. In this section, we have established an algorithm that decides whether a given key set is satisfied by a given relation with missing values. The algorithm works through the given relation as many times as there are keys in the given key set. However, in each scan it does not need to separate tuple pairs that have already been separated in previous runs. Experiments with our algorithm show how the number of keys in a key set lowers the number of offending tuples in real-world benchmark data sets, and how the addition of keys to a key set helps establish entity integrity in data sets with missing values. 

\section{Applications for Automated Reasoning}\label{s:applications}

We illustate some applications of key sets for automated reasoning in databases.  The most important applications of processing data are updates and queries. We briefly describe in this section how automated reasoning about key sets can facilitate each of these application areas.

\subsection{Efficient Updates}

When databases are updated it must be ensured that the resulting database satisfies all the constraints that model the business rules of the underlying application domain. Violations of the constraints indicate sources of inconsistency, and an alert of such inconsistencies should at least be issued to the database administrator. This is to ensure that appropriate actions can be taken, for example, to disallow the update. This quality assurance process incurs an overhead in terms of the time it takes to validate the constraints. As such, users of the database expect that such overheads are minimized. In particular, the time on validating constraints increases with the volume of the database. As a principal, the set of constraints that are specified on the database and therefore subject to validation upon updates, should be non-redundant. That is, no constraints should be specified that are already implied by other specified constraints. The simple reason is that the validation of any implied constraints is a waste of time because the validity of the other constraints already ensures that any implied constraint is valid as well. This is a strong real-life motivation for developing tools that can decide implication. In our running example, the set $\Sigma=\{\mathcal{X}_1,\mathcal{X}_2,\mathcal{X}\}$ of key sets is redundant because the subset $\Sigma'=\{\mathcal{X}_1,\mathcal{X}_2\}$ implies the key set $\mathcal{X}$. Automated solutions to the implication problem can thus automatize the minimization of overheads in validating constraints under database updates.

\subsection{Efficient Queries}

We are interested in the names of patients that can be identified uniquely based on information about their name and the room and time at the accident ward, or based on information about their injury and the time at the accident ward. In SQL, this may be expressed as follows.

\begin{center}
\begin{tabular}{rl}
\texttt{SELECT}   & \emph{name} \\
\texttt{FROM}     & \textsc{ward} \\
\texttt{WHERE}    & \emph{room} \texttt{IS NOT NULL AND} \emph{name} \texttt{IS NOT NULL AND} \\
                  & \emph{time} \texttt{IS NOT NULL} \\
\texttt{GROUP BY} & \emph{room}, \emph{name}, \emph{time} \\
\texttt{HAVING}   & $\texttt{count}(\emph{room},\emph{name},\emph{time})\le 1$ \\
\texttt{UNION}    & \\
\texttt{SELECT}   & \emph{name} \\
\texttt{FROM}     & \textsc{ward} \\
\texttt{WHERE}    & \emph{injury} \texttt{IS NOT NULL AND} \emph{time} \texttt{IS NOT NULL} \\
\texttt{GROUP BY} & \emph{injury}, \emph{time} \\
\texttt{HAVING}   & $\texttt{count}(\emph{injury},\emph{time})\le 1$\;; \\
\end{tabular}
\end{center}

\noindent
Knowing that the underlying relation over \textsc{Ward} satisfies the two key sets $\mathcal{X}_1$ and $\mathcal{X}_2$ and that the key set $\mathcal{X}=\{\{\emph{room},\emph{name},\emph{time}\}, \{\emph{injury},\emph{time}\}\}$ is implied by $\mathcal{X}_1$ and $\mathcal{X}_2$, one can deduce that
every tuple of \textsc{Ward} must be in at least one of the sub-query results of the \texttt{UNION} query. That is, the query above can be simplified to
\begin{center}
\begin{tabular}{rl}
\texttt{SELECT DISTINCT}   & \emph{name} \\
\texttt{FROM}              & \textsc{ward}\;; \\
\end{tabular}
\end{center}
Note that the \texttt{DISTINCT} word is necessary since the \texttt{UNION} operator eliminates duplicates. When evaluated on the example from the introduction, each query will return the result \{(\emph{name}: Miller),(\emph{name}: $\perp$), (\emph{name}: Maier)\}.

Motivated by the applications of key sets for data processing and the lack of knowledge on automated reasoning tasks associated with key sets, the following sections will investigate the implication problem for key sets.

\section{Axiomatizing Key Sets}\label{s:axioms}

\begin{table}[t]
\caption{An axiomatization $\A$ for key sets\label{tab}}
\[\fbox{$\begin{array}{c@{\hspace*{1cm}}c@{\hspace*{1cm}}c}
\cfrac{\calX}{\calX\cup \calY} & \cfrac{\calX\cup\{XY\}}{\calX\cup\{X,Y\}} &\cfrac{\calX_1\quad \calX_2}{\{Z_{(X_1,X_2)}\mid (X_1,X_2) \in \calX_1\times \calX_2\}}\vspace{1ex}\\
&&Z_{(X_1,X_2)}\sub X_1\cup X_2\text{, and }\\
&&X_1 \sub Z_{(X_1,X_2)}\text{ or }X_2 \sub Z_{(X_1,X_2)}\\
\textbf{Upward closure}&\textbf{Refinement}&\textbf{Composition}
\end{array}$}\]
\end{table}

In this section we establish axiomatizations for arbitrary key sets as well as unary ones. This will enable us to effectively enumerate all implied key sets, that is, to determine the semantic closure $\Sigma^\ast=\{\sigma\mid\Sigma\models\sigma\}$ of any given set $\Sigma$ of key sets. A finite axiomatization facilitates human understanding of the interaction of the given constraints, and ensures all opportunities for the use of these constraints in applications can be exploited.

In using an axiomatization we determine the semantic closure by applying \emph{inference rules} of the form $\cfrac{\text{premise}}{\text{conclusion}}$. For a set $\mathfrak{R}$ of inference rules let $\Sigma\vdash_\mathfrak{R}\varphi$ denote the \emph{inference} of $\varphi$ from $\Sigma$ by $\mathfrak{R}$. That is, there is some sequence $\sigma_1,\ldots,\sigma_n$ such that $\sigma_n=\varphi$ and every $\sigma_i$ is an element of $\Sigma$ or is the conclusion that results from an application of an inference rule in $\mathfrak{R}$ to some premises in $\{\sigma_1,\ldots,\sigma_{i-1}\}$. Let $\Sigma^+_\mathfrak{R}=\{\varphi\,\mid\,\Sigma\vdash_{\mathfrak{R}}\varphi\}$ be the \emph{syntactic closure} of $\Sigma$ under inferences by $\mathfrak{R}$. $\mathfrak{R}$ is \emph{sound} (\emph{complete}) if for every set $\Sigma$ over every $R$ we have $\Sigma^+_\mathfrak{R}\subseteq\Sigma^\ast$ ($\Sigma^\ast\subseteq\Sigma^+_\mathfrak{R}$). The (finite) set $\mathfrak{R}$ is a (finite) \emph{axiomatization} if $\mathfrak{R}$ is both sound and complete.

Table~\ref{tab} shows a finite axiomatization $\mathfrak{A}$ for key sets. A non-trivial rule is \textbf{Composition} which is illustrated by our running example.

\begin{example} Recall Example~\ref{ex:intro} from the introduction, in particular $\Sigma=\{\mathcal{X}_1,\mathcal{X}_2\}$ and $\varphi=\mathcal{X}$. It turns out that $\varphi$ is indeed implied by $\Sigma$, since $\varphi$ can be inferred from $\Sigma$ by an application of the Composition rule, and the rule is sound for the implication of key sets. Indeed, $\mathcal{X}_1\times\mathcal{X}_2$ consists of:
\begin{center}
(\{\textit{room},\textit{time}\}, \{\textit{name},\textit{time}\}), \\
(\{\textit{room},\textit{time}\}, \{\textit{injury},\textit{time}\}), \\
(\{\textit{injury},\textit{time}\}, \{\textit{name},\textit{time}\}), and \\
(\{\textit{injury},\textit{time}\}, \{\textit{injury},\textit{time}\})\;.
\end{center}
and for each element $X=(X_1,X_2)$ we need to pick one attribute set $Z_X$ that is contained in the union $X_1\cup X_2$ and contains either $X_1$ or $X_2$. For the first element we pick $\{\textit{room},\textit{time},\textit{name}\}$, and for the remaining three elements we pick $\{\textit{injury},\textit{time}\}$. That results in the key set $\mathcal{X}$.
\end{example}

\begin{comment}
\begin{example}
Let $\{X,Y\}$ and $\{U,V\}$ be two key sets where each $X,Y,U,V$ is a set of two attributes. We use lower case letters with prime to indicate these attributes, e.g., $X=\{x,x'\}$. Then
\[
\cfrac{\{X,Y\}\quad \{U,V\}}{\{Xu,x'V,yU,Yv'\}}
\]
is an instance of  Composition.
\end{example}
\end{comment}

We now proceed with the completeness proof for the axiom system $\A$ of Table \ref{tab}. The proof proceeds in three stages. First in Lemma~\ref{lemA}, we show a characterization of the implication problem. This is applied in Lemma~\ref{lemB} to show that $\A$ extended with  $n$-ary Composition for all $n\in \N$ is complete (see Table~\ref{tab2}). At last, we show in Lemma~\ref{lemC} that $n$-ary Composition  can be simulated with the binary Composition of  $\A$.

\begin{table}[h]
\caption{The $n$-ary Composition rule\label{tab2}}
\[\fbox{$\begin{array}{c@{\hspace*{1cm}}c@{\hspace*{1cm}}c}
 &  \cfrac{\calX_1\quad \ldots\quad \calX_n}{\{Z_{\tuple X}\mid \tuple X \in \calX_1\times \ldots \times \calX_n\}} & \\
&Z_{\tuple X}\sub \bigcup \tuple X\text{ and }\bigvee_i X_i \sub Z_{\tuple X} & \\
%&\textbf{$n$-ary Composition}
\end{array}$}\]
\end{table}

\begin{lemma}\label{lemA}
 $\{\calX_1, \ldots ,\calX_n\}\models \calY$  iff for all $(X_1, \ldots ,X_n)\in \calX_1\times \ldots \times \calX_n$ there is $\calZ\sub \calY$ such that $\bigcup \calZ \sub \bigcup_i X_i $, and $X_i \sub \bigcup \calZ$ for some $i$.
\end{lemma}

\begin{proof}
Assume first that one finds such an $\calZ$. We show that any relation $r$ that satisfies each $\calX_i$ satisfies also $\calY$. Let $t,t'$ be two tuples from $r$. Then for some $(X_1, \ldots ,X_n)\in \calX_1\times \ldots \times \calX_n$, $t$ and $t'$ are both $\bigcup_i X_i$-total and disagreeing on each $X_i$. Assume that $i$ is such that $X_i\sub \bigcup \calZ$, and let $A\in X_i$ be such that $t(A)\neq t'(A)$. Then selecting some $Z\in \calZ$ such that it also contains $A$, we have that $t$ and $t'$ are $Z$-total and deviate on $Z$. Thus $Z$ is witness for $r\models \calY$.

For the other direction we assume that no such $\calZ$ exists. Then there is
$(X_1, \ldots ,X_n)\in \calX_1\times \ldots \times \calX_n$ such that for $\calZ:=\{Z\in \calY\mid Z\sub \bigcup_i X_i\}$, $X_i\not\sub \bigcup \calZ $ for all $i$. Then, selecting an attribute $A_i$ from $X_i\setminus \bigcup \calZ$ for all $i$, we may construct a relation $r$ satisfying $\{\calX_1, \ldots ,\calX_n,\neg  \calY\}$. This relation $r$ consists of two tuples $t,t'$ where $t$ is a constant function mapping all of $R$ to $0$, and $t'$ maps $\bigcup_i A_i$ to $1$, $\bigcup_i X_i \setminus \bigcup_i A_i$ to $0$, and all the remaining attributes to $\bot$. Now, obviously $r$ satisfies all $\calX_i$. Furthermore, for $Y\in \calY\setminus \calZ$, $t'$ is not $Y$-total, and for $Y\in \calY\cap \calZ$ both $t$ and $t'$ are $Y$-total but with constant values $0$. Therefore, $r$ is a witness of $\{\calX_1, \ldots ,\calX_n\}\not\models \calY$ which concludes the proof.
\end{proof}
Notice that the latter condition of Lemma \ref{lemA} can be equivalently stated as $X_i\sub \bigcup\{Y\in \calY\mid Y\sub \bigcup_i X_i\}$ for some $i$.

\begin{lemma}\label{lemB}
The axiomatization $\A$ extended with $n$-ary Composition is complete for key sets.
\end{lemma}
\begin{proof}
Assume $\{\calX_1, \ldots ,\calX_n\}\models \calY$. Then we obtain by Lemma \ref{lemA}  for all $\tuple X=(X_1, \ldots ,X_n)  \in \calX_1\times \ldots \times \calX_n$ a subset $\calZ_{\tuple X}\sub \calY$ such that $\bigcup \calZ_{\tuple X} \sub \bigcup \tuple X$, and $X_i \sub \bigcup \calZ_{\tuple X}$ for some $i$.
Then by Composition we may derive $\{\bigcup \calZ_{\tuple X}\mid \tuple X\in \calX_1\times \ldots \times \calX_n\}$. With repeated applications of Refinement we then derive %(a subset of)
 $\bigcup\{\calZ_{\tuple X}\mid \tuple X \in \calX_1\times \ldots \times \calX_n\}$. Since this set is a subset of $\calY$, we finally obtain $\calY$ with a single application of Upward closure.
\end{proof}

\begin{lemma}\label{lemC}
$n$-ary Composition is derivable in $\A$.
\end{lemma}
\begin{proof}
Assume that $\calK= \{Z_{\tuple X}\mid \tuple X \in \calX_1\times \ldots \times \calX_n\}$ is obtained from $\calX_1,\ldots ,\calX_n$ by an application of $n$-ary Composition.
We will perform consecutive applications of (binary) Composition until we have obtained $\calK$. % If the length of this deduction is maximal, then the key set will be obtained with a single application of Composition to $\bigcup^n_{i=1}\bigcup \calX_i$ and one of $\calX_i$. Furthermore,
 Composition is applied incrementally so that the first application of this rule combines $\calX_1$ and $\calX_2$ to obtain a new key set $\calX$, the second combines $\calX$ and $\calX_{3}$ to obtain the next key set $\calX'$, the third $\calX'$ and $\calX_4$ to obtain $\calX''$, and so forth. Once   $\calX_n$ is reached the cycle is started again from  $\calX_1$.

At each step of the aforementioned procedure we have deduced a key set $\calX$ such that each $X\in \calX$ either is a union $\bigcup  \calY_1 \cup \ldots \cup \bigcup  \calY_n$ for  $\calY_i \sub \calX_i$, or belongs to  the required key set $\calK$. In the previous case, provided that each $\calY_i$ is the maximal subset of $\calX_i$ such that $\bigcup \calY_i \sub X$,  we refer to $\calY_1 \cup \ldots \cup \calY_n$ as the \emph{maximal decomposition} of $X$ and $|\calY_1\cup\ldots \cup \calY_n|$ as the \emph{decomposition size} of $X$. Furthermore, given a set $Z_{\tuple X}\in \calK$ where $\tuple X\in \calX_1\times \ldots \times \calX_n$ we say that a set $X_i\in \tuple X$ is \emph{full} in $Z_{\tuple X}$ if $X_i\sub Z_{\tuple X}$. By the prerequisite of the $n$-ary Composition some member of $\tuple X$ is always guaranteed to be full in $Z_{\tuple X}$.

\noindent
\textbf{Initialization.} Consider an instance of $n$-ary Composition. %Now, each $\calX_i$ is a set of attribute sets of the form $\{X^i_1, \ldots ,X^i_{n_i}\}$.
We initialize the procedure by applying Composition $n-1$ many times so that we obtain the key set $\{\bigcup \tuple X\mid \tuple X \in \calX_1\times \ldots \times \calX_n\}$. This is done by letting $\calU_1:=\calX_1$ and taking the key set $\calU_{i+1}=\{X_1\cup X_2\mid (X_1,X_2) \in \calU_i\times \calX_{i+1}\}$ for $i=1, \ldots ,n-1$.

\noindent
\textbf{Inductive step.} After the initial step we have reached a key set $\calV_1:=\calU_{n}$ such that all $X\in \calV_1\setminus \calK$ have decomposition size at least $1$. Assume now that we have reached a key set $\calV_{m}$ such that all $X\in \calV_m\setminus \calK$ have decomposition size at least $m$. As the induction step we show how to obtain a key set $\calV_{m+1}$ such that every member of  $\calV_{m+1}\setminus \calK$ has decomposition size at least $m+1$. This is done by taking a single round of applications of Composition to $\calV_m$ and $\calX_1, \ldots ,\calX_n$. That is, $\calV_m$ and $\calX_1$ are first combined using  Composition, then the outcome is combined with $\calX_2$, and its outcome with $\calX_3$, and so forth until we have applied this procedure to $\calX_n$. All these applications keep the members of $V_m \cap \calK$ fixed. For instance, at the first step $Z_{(X,Y)}$ for $X\in \calV_m\cap \calK$ and any $Y\in \calX_1$ is defined as $X$. We show how this deduction handles an arbitrary $X\in \calV_m\setminus \calK$.

By induction assumption each $X\in \calV_m\setminus \calK$ has decomposition size at least $m$. Let  $\bigcup  \calY_1  \cup \ldots \cup \bigcup \calY_n$ be the maximal decomposition of $X$. Now, assume towards a contradiction that for each $i$ there is $Y_i\in \calY_i$ such that $Y_i$ is not full in any $Z_{\tuple Y}\in \calK$ where $\tuple Y\in \calY_1\times \ldots \times  \calY_n$ and $Y_i$ is the $i$th member of $\tuple Y$. Then, however,  the diagonal  $\tuple Y'=(Y_1, \ldots ,Y_n)$ must have a member that is full in $Z_{\tuple Y'}$. This is a contradiction and hence there is $i$ such that all $Y_i\in \calY_i$ are full in some $Z_{\tuple Y}\in \calK$ where $\tuple Y\in \calY_1\times \ldots \times \calY_n$ and $Y_i$ is the $i$th member of $\tuple Y$. With regards to $X$, Composition is then applied as follows. For the first $i-1$ applications $X$ is kept fixed. For the $i$th application that considers $\calX_i$,  each pair of $X$ and $Y\in \calY_i$ is transformed to that $Z_{\tuple Y}\in \calK$ in which $Y$ is full. Furthermore, each pair of $X$ and $Y\in \calX_i\setminus \calY_i$ is transformed to $XY$. Take note that the decomposition size of $XY$ is at least $n+1$. At last, the remaining applications of Composition keep the obtained sets fixed. Since this procedure is applied to all $X\in \calV_m\setminus \calK$, we obtain that $\calV_{m+1}\setminus \calK$ has only sets with decomposition size at least $m+1$. This concludes the induction step.

Now, $\calV_{M+1}$ where $M=|\calX_1 \cup \ldots \cup \calX_n|$ is a subset of $\calK$. Hence, we conclude that   $\calV_{M+1}$ yields $\cal{K}$ with one application of Upward closure.
\end{proof}
Note that a simulation of one application of $n$-ary of Composition to $\{\calX_1, \ldots ,\calX_n\}$ takes at most $(n+1)\cdot |\bigcup_{i=1}^n \calX_i|$ applications of binary Composition plus one application of Upward Closure.

The previous three lemmata now generate the following axiomatic characterization of key set implication. We omit the soundness proof which is straightforward to check.
\begin{theorem}
The axiomatization $\A$ is sound and complete for key sets.
\end{theorem}

\noindent
\textbf{Another important application.} A direct application of an axiomatization is the efficient representation of collections of key sets. Similar to the computation of non-redundant covers during update operations, removing any redundant constraints makes the result easier to understand by humans. This is, for example, important for the
discovery problem of key sets in which one attempts to efficiently represent all those key sets that a given relation satisfies. Even more directly, one can understand any
sound inference rule as an opportunity to apply pruning techniques as part of a discovery algorithm. A complete axiomatization ensures all opportunities for the pruning of a search space can be exploited.

\section{Complexity of Key Set Implication}\label{s:complexity}

In this section we settle the exact computational complexity of the implication problem for key sets. While the implication problem for most notions of keys over incomplete relations is decidable in linear time, the implication problem for key sets is likely to be intractable. This should also be seen as evidence for the expressivity of key sets.

\begin{theorem}
The implication problem for key sets is coNP-complete.
\end{theorem}

\begin{proof}
Consider first the membership in $\coNP$. By Lemma \ref{lemA}, for determining whether $\{\calX_1, \ldots ,\calX_n\}\not\models \calY$, it suffices to choose $X_1, \ldots ,X_n$ respectively from $\calX_1, \ldots ,\calX_n$, and then deterministically check that  $X_i \not\sub \bigcup \calZ$ for all $i$, where $\calZ$ is selected deterministically as $\calZ:=\{Z\in \calY\mid Z\sub \bigcup_i X_i\}$.
%We only note that $\calZ$ can be selected deterministically by taking the set of all $Z\in \calY$ such that $Z\sub \bigcup_i X_i$.

\noindent
For the hardness, we reduce from the complement of \threeSAT. Let $C_1, \ldots ,C_n$ be a collection of clauses, each consisting of three literals, i.e., propositions of the form $p$ or negated propositions of the form $\neg p$. Let $P$ be the set of all proposition symbols that appear in some $C_i$, and let $\overbar{P}$ consist of their negations. Letting $P\cup \overbar{P}$ be our relation schema, we show that $\bigwedge_i \bigvee C_i$ has a solution iff $\{\{p,\neg p\}\mid p\in P\} \not\models \{C_1, \ldots ,C_n\}$. Notice that the antecedent is a set of singleton key sets, each of size two.

Assume first that there is a solution. Let $S\sub \calP(P)$ encode the complement of that solution, i.e., $S$ is such that each $C_i$ contains some $p\notin S$ or some $\neg p$ for $p\in S$. Let $\overbar{S}=\{\neg p \mid p\not\in S\}$, and define singleton sets $X_p=\{p,\neg p\} \cap (S\cup\overbar{S})$, encoding those literals that are set false by the solution. Then $C_i\not\sub \bigcup_p X_p$ for all $i$, implying be Lemma \ref{lemA} that $\{\{p,\neg p\}\mid p\in P\} \not\models \{C_1, \ldots ,C_n\}$.

Assume then that $\{\{p,\neg p\}\mid p\in P\} \not\models \{C_1, \ldots ,C_n\}$. By Lemma \ref{lemA} we find $X_p\in \{p,\neg p\}$ such that for no $\calZ \sub \{C_1, \ldots ,C_n\}$ we have that
$\bigcup \calZ \sub \bigcup_p X_p$ and $\bigvee_pX_p\sub \bigcup \calZ$. Now, $C_i\sub \bigcup_p X_p$ implies $X_p\sub C_i$ for three distinct $p$, and therefore we must have $C_i\not\sub \bigcup_p X_p$ for all $i$. It is now easy to see that the sets $X_p$ give rise to a solution to the satisfiability problem.
\end{proof}

\section{Logical Characterization of the Implication Problem}\label{s:logical}

\section{Armstrong Relations}\label{s:Armstrong}

In this section we ask the basic question whether key sets enjoy Armstrong relations. These are special models which are perfect for a given collection of key sets. More formally, a given relation $r$ is said to be \emph{Armstrong} for a given set $\Sigma$ of key sets if and only if for all key sets $\varphi$ it is true that $r$ satisfies $\varphi$ if and only if $\Sigma$ implies $\varphi$. Indeed, an Armstrong relation is a perfect model for $\Sigma$ since it satisfies all keys sets implied by $\Sigma$ and does not satisfy any key set that is not implied by $\Sigma$. Armstrong relations have important applications in data profiling \cite{DBLP:journals/vldb/AbedjanGN15} and the requirements acquisition phase of database design \cite{DBLP:journals/is/LangeveldtL10}.

Unfortunately, arbitrary sets of key sets do not enjoy Armstrong relations as the following result manifests.

\begin{theorem}
There are sets of key sets for which no Armstrong relations exist.
\end{theorem}

\begin{proof}
An example is \[\Sigma=\{\{\{A\},\{B\}\},\{\{C\},\{D\}\}\}\] with attributes $A,B,C,D$. Then $\sigma_1=\{\{A,C\},\{A,D\},\{B,C\}\}$ and $\sigma_2=\{\{A,D\},\{B,C\},\{B,D\}\}$ are two non-consequences of $\Sigma$, respectively exemplified by the two 2-tuple relations on the left of Figure~\ref{figA}, where ``$d$'' refers to any distinct total value.

\begin{figure}[t]
\begin{center}
\begin{tabular}{cccc}\toprule
$A$&$B$&$C$&$D$\\\midrule
$d$&$d$&$\bot$&$d$\\
$\bot$&$d$&$d$&$d$\\\bottomrule
\end{tabular}%\\[2ex]
\qquad
\begin{tabular}{cccc}\toprule
$A$&$B$&$C$&$D$\\\midrule
$d$&$d$&$d$&$\bot$\\
$d$&$\bot$&$d$&$d$\\\bottomrule
\end{tabular}%\\[2ex]
\qquad
\begin{tabular}{cccc}\toprule
$A$&$B$&$C$&$D$\\\midrule
$d$&$d$&$\bot$&$d$\\
$\bot$&$d$&$d$&$d$\\
$d$&$d$&$d$&$\bot$\\
$d$&$\bot$&$d$&$d$\\\bottomrule
\end{tabular}%\\[2ex]
\end{center}
\caption{\label{figA}}
\vspace{-4mm}
\end{figure}

These are the only possible types of tuple pairs that satisfy $\Sigma\cup\{\neg \sigma_1\}$ and $\Sigma\cup\{\neg\sigma_2\}$, respectively.
Therefore, we observe that any relation $r$ satisfying $\Sigma$ and refuting both $\sigma_1$ and $\sigma_2$ has a homomorphism from a relation of the form on the right of Figure~\ref{figA} to a subset of $r$ with the condition that this homomorphism preserves nulls and maps domain values to domain values. However, then neither $\{\{A\},\{B\}\}$ nor $\{\{C\},\{D\}\}$ is a key set anymore.
\end{proof}

\section{Implication for Unary by Arbitrary Key Sets}\label{s:fragment}

In this section we identify a fragment of key sets for which automated reasoning is efficient. This is strongly motivated by the results of the previous sections in which the coNP-completeness of the implication problem, and the lack of general Armstrong relations has been established. Indeed, the fragment is the implication of unary key sets by arbitrary key sets. We show that this fragment is captured axiomatically by the Refinement and Upward Closure rules, can be decided in time quadratic in the input, and Armstrong relations always exist and can be computed with conservative use of time and space.

\subsection{An algorithmic characterization}

Our first result establishes that unary key sets must be implied by a single key set from the given collection of key sets.

\begin{theorem}\label{t:algorithmic}
Let $\Sigma=\{\calX_1, \ldots ,\calX_n\}$ be a collection of arbitrary key sets, and let $\varphi=\{\{A_1\},\ldots,\{A_k\}\}$ be a unary key set over relation schema $R$. Then $\Sigma$ implies $\varphi$ if and only if there is some $i\in\{1,\ldots,n\}$ such that $ \bigcup \calX_i\sub \{A_1,\ldots,A_k\}$.
\end{theorem}

\begin{proof} If $\bigcup \calX_i\sub X$ for some $i\in\{1,\ldots,n\}$, Refinement and Upward Closure infer $\varphi$ from $\Sigma$. Due to the rules' soundness, $\varphi$ is implied by $\Sigma$.

Vice versa, assume that $\bigcup \calX_i\not\sub X$ holds for all $i=1,\ldots,n$. Let $r$ be defined as $r=\{t,t'\}$ where $t$ and $t'$ are two total tuples that agree on $X=\{A_1,\ldots,A_k\}$ and disagree elsewhere. It follows that $r$ violates $\varphi$. Since $\bigcup \calX_i\not\sub X$ for all $i=1,\ldots,n$, $t_1$ and $t_2$ must differ on some attribute in $\bigcup \calX_i$ for $i=1,\ldots,n$. This means, $r$ satisfies all key sets in $\Sigma$.  Consequently, $\Sigma$ does not imply $\varphi$.
\end{proof}

\noindent
A direct consequence of Theorem~\ref{t:algorithmic} is the quadratic time complexity of the implication problem for unary by arbitrary key sets. For a collection $\Sigma$ of key sets let $|\Sigma|$ denote the total number of attribute occurrences in elements of $\Sigma$.

\begin{corollary}
The implication problem of unary key sets by arbitrary key sets is decidable in time $\mathcal{O}(|\Sigma|\times|\varphi|)$ in the input $\Sigma\cup\{\varphi\}$.
\end{corollary}

\subsection{A finite axiomatization}

Our next result establishes a finite axiomatization for the implication of unary by arbitrary key sets that consists of the Refinement and Upward Closure rules. As this fragment is decidable in time quadratic in the input, and the general case is \textit{coNP}-complete, the Composition rule is the source of likely intractability.

\begin{corollary}
The implication problem of unary key sets by arbitrary key sets has a sound and complete axiomatization in Refinement and Upward Closure.
\end{corollary}

\begin{proof} Let $\Sigma=\{\calX_1, \ldots ,\calX_n\}$ be a set of key sets, and let $\varphi=\{\{A_1\},\ldots,\{A_k\}\}$ be a unary key set over relation schema $R$. If $\varphi$ can be inferred from $\Sigma$ by a sequence of applications of the Refinement and Upward Closure rules, the soundness of these rules ensures that $\varphi$ is also implied by $\Sigma$.

For completeness we assume that $\varphi$ cannot be inferred from $\Sigma$ by means of applications using the Refinement and Upward Closure rules. Hence, $\bigcup \calX_i\not\sub X$ holds for all $i=1,\ldots,n$. Theorem~\ref{t:algorithmic} shows that $\Sigma$ does not imply $\varphi$.
\end{proof}

\subsection{Existence and computation of Armstrong relations}

Armstrong models relative to unary consequences are also easy to obtain. It merely suffices to take a disjoint union of all of the two tuple relations mentioned in the proof of Theorem~\ref{t:algorithmic}.

\begin{corollary}
The implication problem of unary key sets by arbitrary key sets has Armstrong relations.
\end{corollary}

While the existence of perfect models is easy to come by the disjoint union construction, an actual generation of Armstrong relations by this construction is not efficient. Smaller Armstrong relations can be constructed as follows. Theorem~\ref{t:algorithmic} shows that the implication problem of unary key sets $\mathcal{X}$ by a collection $\Sigma=\{\calX_1, \ldots ,\calX_n\}$ of arbitrary key sets only depends on the attributes contained in each given key set of $\Sigma$, and not on how they are grouped as sets in a key set. We thus identify, without loss of generality, $\mathcal{X}$ with $\bigcup\mathcal{X}$ and each $\mathcal{X}_i$ with $\bigcup\mathcal{X}_i$.

The idea is then to compute so-called anti-keys, which are the maximal subsets of the underlying relation schema which are key sets not implied by $\Sigma$. Given the anti-keys, an Armstrong relation for $\Sigma$ can be generated by starting with a single complete tuple, and introducing for each anti-key a new tuple that has matching total values on the attributes of the anti-key and unique values on attributes outside the anti-key. This construction ensures that all non-implied (unary) key sets are violated and all given key sets are satisfied. The computation of the anti-keys from $\Sigma$ can be done by taking the complements of the minimum transversals of the hypergraph formed by the elements of $\Sigma$. A transversal for a given set of attribute subsets $\mathcal{X}_i$ is an attribute subset $\mathcal{T}$ such that $\mathcal{T}\cap\mathcal{X}_i\not=\emptyset$ holds for all $i$. While many efficient algorithms exist for the computation of all hypergraph transversals, it is still an open problem whether there is an algorithm that is polynomial in the output \cite{DBLP:journals/siamcomp/EiterGM03}. We can show that this construction always generates an Armstrong relation whose number of tuples is at most quadratic in that of an Armstrong relation that requires a minimum number of tuples.

\begin{corollary}
Armstrong relations that are at most quadratic in that of a minimum Armstrong relation can be generated for unary by arbitrary key sets.
\end{corollary}

\begin{proof}[Sketch]
One can show first that a given relation is Armstrong for a given set of key sets if and only if for every anti-key the relation has two tuples which have matching values on exactly those attributes that form the anti-key and for no union over the elements of a key set there is a pair of tuples with matching values on all attributes in the union. Subsequently, one can show that the number of tuples in a minimum-sized Armstrong relation is bounded from below by one half of the square root of 1 plus 8 times the number of anti-keys, and bounded upwards by the increment of the number of anti-keys. Consequently, our construction generates an Armstrong relation that is at most quadratic in a minimum-sized Armstrong relation.
\end{proof}

Our construction can also be viewed as a construction of Armstrong relations for certain keys by key sets. Note that \cite{DBLP:journals/vldb/KohlerLLZ16} constructed Armstrong relations for sets of possible and certain keys under \texttt{NOT NULL} constraints, whenever they exist. Our construction here does not require null markers.

\begin{example} Consider the set $\Sigma=\{\mathcal{X}_1,\mathcal{X}_2\}$ with $\mathcal{X}_1$ and $\mathcal{X}_2$ from Example~\ref{ex:intro-2} over the relation schema \textsc{Ward}. Then $\bigcup\mathcal{X}_1=\{\textit{room},\textit{time},\textit{injury}\}$ and $\bigcup\mathcal{X}_2=\{\textit{name},\textit{time},\textit{injury}\}$. The minimum transversals would be $\mathcal{T}_1=\{\textit{time}\}$, $\mathcal{T}_2=\{\textit{injury}\}$, and $\mathcal{T}_3=\{\textit{room},\textit{name}\}$, and their complements on \textsc{Ward} are the anti-keys \[\mathcal{A}_1=\{\textit{room},\textit{name},\textit{address},\textit{injury}\},\] \[\mathcal{A}_2=\{\textit{room},\textit{name},\textit{address},\textit{time}\}, \text{and}\] \[\mathcal{A}_3=\{\textit{address},\textit{injury},\textit{time}\}.\] The following relation is Armstrong for $\Sigma$.
\begin{center}
\begin{tabular}{c@{\hspace*{.5cm}}c@{\hspace*{.5cm}}c@{\hspace*{.5cm}}c@{\hspace*{.5cm}}c}\hline
\textit{room} & \textit{name} & \textit{address} & \textit{injury} & \textit{time} \\ \hline
1             & Miller        & 24 Queen St      & leg fracture    & Sunday, 16 \\
1             & Miller        & 24 Queen St      & leg fracture    & Monday, 19 \\
1             & Miller        & 24 Queen St      & arm fracture    & Monday, 19 \\
2             & Maier         & 24 Queen St      & arm fracture    & Monday, 19 \\ \hline
\end{tabular}
\end{center}
The relation satisfies $\mathcal{X}_1$ and $\mathcal{X}_2$, but the relation violates the unary key set $\varphi'=\{\{\textit{room}\},\{\textit{name}\}, \{\textit{address}\},\{\textit{time}\}\}$, so $\varphi'$ is not implied by $\Sigma$.
\end{example}

\section{Conclusion and Future Work}\label{s:conclusion}

We took first steps in investigating limits and opportunities for automated reasoning about key sets in databases. Key sets provide a more general and flexible implementation of entity integrity than Codd's notion of a primary key. We established a linear-time algorithm for the validation of a given key set on a given data set, and experimentally demonstrated its runtime behavior. We showed that the implication problem for general key sets enjoys a binary axiomatization, is \emph{coNP}-complete, and lacks Armstrong relations. The implication problem of unary key sets by arbitrary key sets enjoys a unary axiomatization, is decidable in quadratic input time, and Armstrong relations can always be generated using hypergraph transversals such that the number of tuples is guaranteed to be at most quadratic in the minimum number of tuples required. Our results provide a foundation for controlling entity integrity in databases with missing values.

\noindent
Interesting questions arise in theory and practice. Our \textit{coNP}-completeness result calls for fixed-parameter solutions. A characterization for the existence of Armstrong relations in the general case would be interesting, and their efficient construction whenever possible. The validation of key sets in databases is an important practical issue, for which effective index structures need to be found. The problem of computing all key sets that hold in a given relation is important for data profiling \cite{DBLP:journals/vldb/AbedjanGN15}. Automated reasoning about foreign key sets is interesting as they generalize referential integrity \cite{DBLP:journals/ita/LeveneL01}. Similar to how functional and inclusion dependencies and independence atoms interact \cite{DBLP:journals/jcss/CasanovaFP84,DBLP:journals/jcss/KohlerL17}, automated reasoning for functional, multivalued, and inclusion dependency sets is interesting \cite{DBLP:journals/jcss/HannulaKL16}.

\bibliographystyle{plain}
\bibliography{biblio}
\end{document}